\documentclass[aps,reprint,superscriptaddress,floatfix,prl]{revtex4-1}

\usepackage{graphicx}
    \graphicspath{{./Images/}}
\usepackage{xcolor}
\usepackage[colorlinks = true, citecolor=blue, linkcolor=black]{hyperref}
\usepackage{enumitem}
    \setlist[itemize]{noitemsep, topsep=0pt}
    \setlist[enumerate]{noitemsep, topsep=0pt}
\usepackage{verbatim}
\usepackage{mathtools,amssymb,amsmath}
\usepackage{bm}
\usepackage{yhmath}
\usepackage[inline]{asymptote}
\usepackage{cancel}
\usepackage{relsize}
\usepackage{array}
\usepackage{bm}

\newcommand{\be}{\begin{equation}}
\newcommand{\ee}{\end{equation}}
\newcommand{\<}{\langle}
\renewcommand{\>}{\rangle}

\newcommand{\Tr}{{\rm Tr\,}}

\usepackage{color}

\newcommand{\limsum}{\lim_{N,D\to\infty}\frac 1D\sum_{n=0}^{D-1}}

\usepackage{amsthm}

\newtheorem{prop}{Proposition}

\newtheorem{lem}{Lemma}

\newtheorem{dfn}{Definition}

\newcommand{\ourtitle}{Observing and braiding topological Majorana modes\\ on programmable quantum simulators}

\newcommand{\almaden}{IBM Quantum, IBM Research -- Almaden, San Jose CA, 95120, USA}
\newcommand{\cambridge}{IBM Quantum, MIT-IBM Watson AI lab,  Cambridge MA, 02142, USA}
\newcommand{\yale}{Department of Physics, Yale University, New Haven CT, 06520, USA}
\newcommand{\google}{Google Quantum AI, Venice Beach, CA, 90291, USA}

\newcommand{\sectionone}{Section 1}
\newcommand{\sectiononeexpl}{SI Section 1}
\newcommand{\sectiononebig}{Section 1 of Supplementary Information (SI)}
\newcommand{\sectiontwo}{Section 2}

\newcommand{\sectiontwobig}{SI Section 2}
\newcommand{\sectionthree}{Section 3}
\newcommand{\sectionthreeexpl}{SI Section 3}
\newcommand{\sectionfour}{Section 4}
\newcommand{\sectionfourexpl}{SI Section 4}
\newcommand{\sectionfive}{Section 5}
\newcommand{\sectionfiveexpl}{SI Section 5}
\newcommand{\figsize}{0.5}
\newcommand{\results}{}
\newcommand{\introductionsec}{}
\newcommand{\discussion}{\textbf{Discussions and outlook.}}

\newcommand{\insupp}{in the Supplementary Information}

\newcommand{\seesupp}{(see discussion below Eq.~\eqref{eqs:wavefunctions_F})}

\begin{document}

\title{\ourtitle}

\author{Nikhil Harle}
\affiliation{\yale}
\affiliation{\cambridge}
\author{Oles Shtanko}
\affiliation{\almaden}
\author{Ramis Movassagh}
\affiliation{\cambridge}
\affiliation{\google}

\begin{abstract}
Electrons are indivisible elementary particles, yet paradoxically a collection of them can act as a fraction of a single electron, exhibiting exotic and useful properties. One such collective excitation, known as a topological Majorana mode, is naturally stable against perturbations, such as unwanted local noise, and can thereby robustly store quantum information. As such, 
Majorana modes serve as the basic primitive of topological quantum computing, providing resilience to errors. However, their demonstration on quantum hardware has remained elusive. Here, we demonstrate a verifiable identification and braiding of topological Majorana modes using a superconducting quantum processor as a quantum simulator. By simulating fermions on a one-dimensional lattice subject to a periodic drive, we confirm the existence of Majorana modes localized at the edges, and distinguish them from other trivial modes. To simulate a basic logical operation of topological quantum computing known as braiding, we propose a non-adiabatic technique, whose implementation reveals correct braiding statistics in our experiments. This work could further be used to study topological models of matter using circuit-based simulations, and shows that long-sought quantum phenomena can be realized by anyone in cloud-run quantum simulations, whereby accelerating fundamental discoveries in quantum science and technology.
\end{abstract}

\maketitle

\introductionsec

It is a unique time in the history of science and engineering when we are witnessing significant advances in the development of fully controllable, coherent many-body quantum systems that contain dozens to hundreds of qubits~\cite{chow2021ibm}. Quantum simulators hold the promise of exponentially outperforming classical computers, which would bring about a host of applications beyond the reach of classical computers. Perhaps the most promising  application of these systems is the simulation of quantum many-body systems~\cite{feynman1982simulating}, which includes topological phases of matter \cite{wen2017zoo,qi2011topological}. In addition to their exotic nature, topological quantum states are a promising route to fault-tolerant quantum computation that is based on non-Abelian excitations such as Majorana fermions \cite{aasen2016milestones}. Majorana fermions are exotic particles: each is its own antiparticle, unlike an electron being distinct from its antiparticle (positron). Despite the remarkable progress, the original proposal for the realization of Majorana-based quantum memories on solid state devices \cite{lutchyn2010majorana,oreg2010helical,beenakker2013search} ultimately encountered difficulties due to disorder and lack of control, as well as the inability to separate Majorana modes from other trivial zero-energy states~\cite{lee2014spin,kayyalha2020absence,valentini2021nontopological,yu2021non,saldana2021coulombic,wang2021spin}. At the same time, quantum simulators may help in this search with their unprecedented levels of parameter control for a range of topological models \cite{khemani2016phase,liu2013floquet,potter2016classification}.

Realization of topological phases hosting Majorana modes in bosonic multi-qubit devices was first envisioned few decades ago \cite{levitov2001quantum}, with subsequent theoretical developments \cite{you2014encoding,backens2017emulating}. 
Since then signatures of topological modes were detected in photonic experiments \cite{kitagawa2012observation,cheng2019observation,xiao2017observation} and programmable digital quantum information processors \cite{smith2019simulating,tan2019simulation,fauseweh2021digital,bassman2021simulating,koh2021stabilizing,smith2019crossing,neil2021accurately}. While these devices are limited to non-equilibrium settings 
they still are able to exhibit long-lived signatures of topological modes \cite{wikeckowski2018identification,shtanko2020unitary}. Some of these signatures were analyzed in programmable processors with methods usually tailored to free-fermionic models \cite{azses2020identification,choo2018measurement,zhang2022digital}. However, the qualitative study of the properties of these topological excitations remained a challenge.
Braiding of the Majorana fermions is yet another motivation as  it provides the exchange statistics of the topological excitations and is a necessary step for topological quantum computation. While there has been progress in manipulation of toy Majorana modes in photonics~\cite{xu2016simulating,xu2018photonic,liu2021topological} and superconducting architectures~\cite{wootton2017demonstrating,zhong2016emulating,song2018demonstration,huang2021emulating},  they were limited to a few qubit systems, not a real topological phase. Thus, direct probing of the topological modes and their manipulation remained an open problem.

Using existing noisy quantum hardware, we aim to perform quantitative simulations of topological quantum matter. We recreate the state of one-dimensional topological superconductor widely known for hosting a pair of exotic ``half-electron"  Majorana modes at its boundaries. We show how to use Fourier transformation of multi-qubit observables to reliably determine the structure of Majorana modes. We also demonstrate how the detection of two-point correlation functions make it possible to distinguish between trivial and topological modes. Finally, we introduce and implement the fast approximate swap (FAS): a general non-adiabatic method to approximately braid Majorana fermions in one dimension. Unlike conventional adiabatic methods, it allows implementation on
the current generation of noisy quantum hardware.\\

\results

\textbf{Floquet engineering}.
Time-periodic (Floquet) systems  had proven to be particularly suitable
for  simulations
on digital quantum processors. 
In particular, when system Hamiltonian alternates between two or more local Hamiltonians being sums of mutually commuting terms,
this choice of quantum dynamics provides a remarkable resource utilization. In this way, unlike trotterized continuous dynamics, a constant-time Floquet dynamics can be simulated using constant-depth circuit.
While Floquet systems may be compared in their form to rough trotterization of continuous dynamics, they exhibit a wide variety of topological phases \cite{rudner2020band}. The Floquet topological phase may be quite robust despite the presence of disorder~\cite{shtanko2018stability}.

Our focus is on the time-periodic Hamiltonian
\be\label{main_ham}
H(t) = \sum_{j=1}^{N-1}\Bigl(J(t)\,X_j X_{j+1}+\lambda(t)\, Z_jZ_{j+1}\Bigl)+h(t)\sum_{j=1}^{N} Z_j ,
\ee
where $X_j$ and $Z_j$ are single-qubit Pauli operators, $\{J,\lambda,h\}(t+T)=\{J,\lambda,h\}(t)$ is a set of time-periodic parameters, $T$ is the time period, $N$ is the number of qubits.

We propose a protocol that divides a single driving period into three parts.
For simplicity, we consider the driving period acting from $t=0$ to $t=T$.
During the first part, from the start of the period to time $\tau_1$, we set $h(t) = h$ and the other coefficients to zero, $J(t)=\lambda(t) = 0$.
Next, for times in between $\tau_1$ and $\tau_2$, we set $J(t) = J$ and the rest of the coefficients to zero.
Lastly, between $\tau_2$ and the end of the period $T$, we set the last term to be on, $\lambda(t) = \lambda$, and all other terms to zero.
Therefore, only one term in the Hamiltonian in Eq.~\eqref{main_ham} is active at any given moment.

A quantum circuit can reproduce such a quantum dynamics protocol at discrete times $t_n = nT$. At such times, the system's state is described by the wavefunction $|\psi_n\> = U_F^n|\psi_0\>$, here $|\psi_0\>$ is the initial state and $U_F=\exp(-i\int_0^t H(t') dt')$ is the Floquet unitary,
\be\label{eq:floquet_unitary}
\begin{split}
 \quad U_F = \prod_{j=1}^{N-1}e^{-i\varphi Z_jZ_{j+1}}&\prod_{j=1}^{N-1}e^{-i\theta X_jX_{j+1}}
 \prod_{j=1}^N e^{-i\phi Z_j},
 \end{split}
\ee
where the gate angles are $\phi = h\tau_1$, $\theta = J(\tau_2-\tau_1)$, and $\varphi = \lambda (T-\tau_2)$. The corresponding experimental protocol that involves local single- and two-qubit gates is depicted in Fig.~\ref{fig:fig1}(a), where each cycle corresponds to a single Floquet unitary. 

\begin{figure}[t!]
    \centering
    \includegraphics[width=\figsize\textwidth]{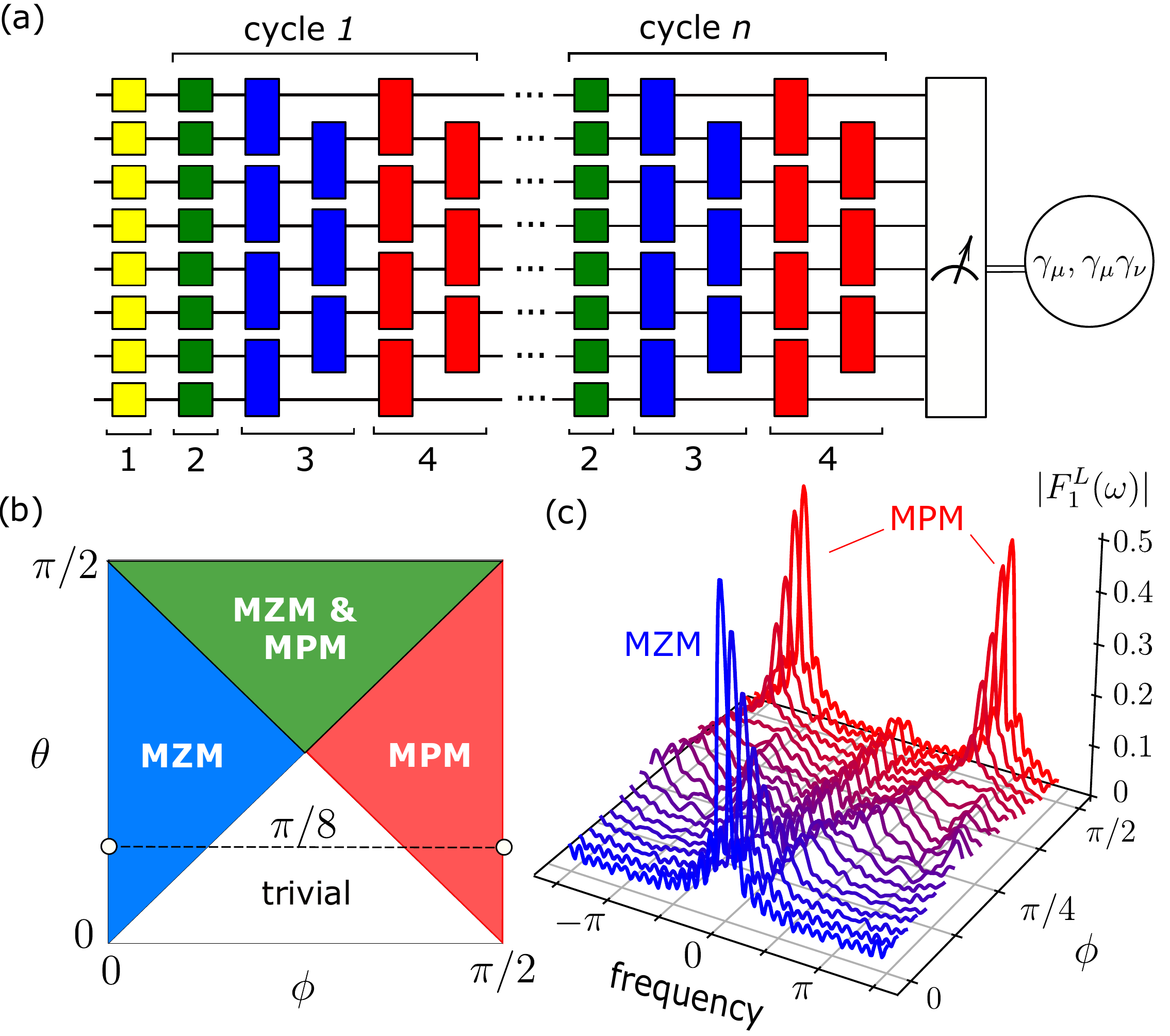}
    \caption{\textbf{Circuit and phase diagram.} (a) Schematics for an 8-qubit circuit including the initialization, evolution, and measurement parts. The initialization process involves single-qubit Hadamard gates (yellow, \#1). Evolution is composed of cycles consisting of $Z$-gates (green, \#2), $XX$-gates (blue, \#3), and $ZZ$-gates (red, \#4). Measurement provides the expectation of operators $\gamma_\mu^{L,R}$ or $\gamma_\mu\gamma_\nu$ (see \sectiononeexpl). (b) Phase diagram for $\lambda=0$, depicting four possible phases, see text. (c) Experimentally measured Fourier component $|F^L_1(\omega)|$ as a function of $\phi$ for fixed $\theta = \pi/8$ using a 21-qubit system, implemented on \textit{ibm\_hanoi}. The system exhibits transitions from MZM to trivial phase and from trivial to MPM phase. Detected peaks indicate the presence of Majorana modes at frequencies $\omega = 0$ and $\omega = \pi$. }
    \label{fig:fig1}
\end{figure}

The model has received considerable attention in the study of condensed matter systems due to its alternative description in terms of spinless fermions. By Jordan-Wigner transformation, the qubit Pauli operators can be transformed into nonlocal Majorana fermion operators $\gamma_\mu$ satisfying $\{\gamma_\mu,\gamma_\nu\} = 2\delta_{\mu\nu}$, where $\mu,\nu=1,\dots,2N$ \cite{suzuki2012quantum}. It is not a unique mapping; here we use two equivalent Jordan-Wigner representations, denoted as $\gamma^{L,R}_\mu$ and associated with the right and left boundaries. In these representations a Majorana operator becomes a string of Pauli operators connected to one of the boundaries. As we show in \sectiononebig, the expectation values of these operators can be obtained from single-qubit measurements preceded by a series of two-qubit gates.
 We will not include the superscripts for the Majorana operators when the choice of representation is not important. 

\begin{figure*}[t!]
    \centering
    \includegraphics[width=1\textwidth]{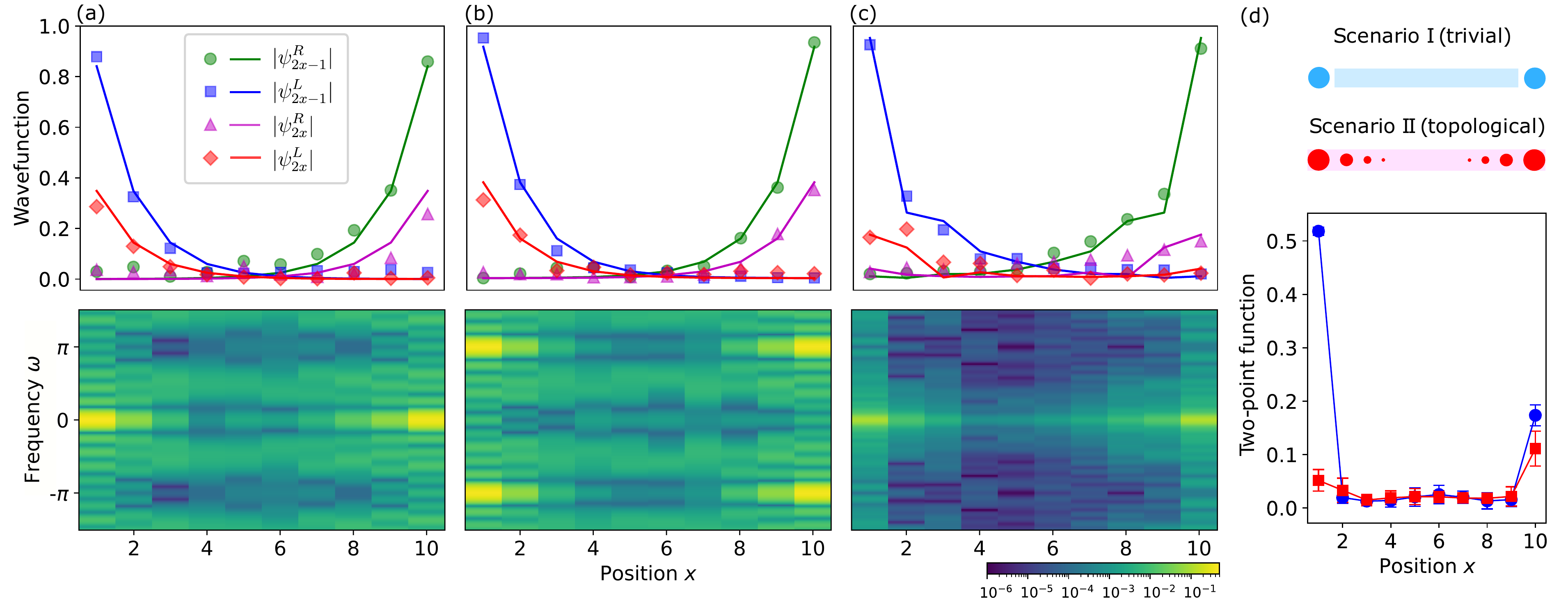}
    \caption{\textbf{Detection of Majorana modes}. 
    Panels (a)-(c) show the absolute value of the experimentally observed Majorana mode wavefunctions $|\psi^s_\mu|$ (dots) in comparison with its theoretical prediction (lines) for $N=10$ qubits. Wavefunctions are further normalized because under noise effects Eq.~\eqref{eq:wavefnct_extract} is inexact. Bottom panel illustrates the density function $g(x,\omega) = |F^L_{2x-1}(\omega)|^2+|F^R_{2x-1}(\omega)|^2$, the bright peaks show the frequency of the modes. (a) MZM extracted using \textit{ibm\_montreal} device in the topological phase $\theta = \pi/4$, $\phi = \pi/8$, and $\varphi = 0$, using $D=11$ cycles (b) MPM extracted using the same device in topological phase $\theta = \pi/4$, $\phi = 3\pi/8$, and $\varphi = 0$, using $D=11$ cycles (c) MZM wavefunction extracted using \textit{ibm\_mumbai} device for interacting topological phase $\theta = \pi/4$, $\phi = \pi/16$, and $\varphi = \pi/16$, using $D=21$ cycles. (d) Difference between trivial phase $\theta = \pi/16$ and $\phi = \pi/4$ with two trivial boundary modes (blue, circles) and topological phase $\theta = \pi/4$ and $\phi = \pi/16$ (red, squares) quantified by $|T_{1,2x}|$ in Eq.~\eqref{eq:t_observable}, measured using \textit{ibm\_toronto} device. The result is calculated as the average of 10 random initial states and $D=11$ cycles. The error bars are one standard deviation. The expectation values used to generate all figures are calculated by averaging over 8192 circuit runs. }
    \label{fig:fig2}
\end{figure*}

In the case that $\lambda = 0$, the Hamiltonian in Eq.~\eqref{main_ham} is non-interacting and takes the simple quadratic form $H(t) = \sum_{\mu,\nu=1}^{2N} h_{\mu\nu}(t)\gamma_\mu\gamma_\nu$, where $h_{\mu\nu}$ is an antisymmetric Hermitian matrix. Due to its free fermionic nature, dynamics generated by such a Hamiltonian are classically  efficient to simulate (see \sectiontwobig). 
In this regime, depending on the ratio between $J$ and $h$, the system 
exhibits 
various phases including the symmetry-protected topological phases \cite{khemani2016phase,potter2016classification}, as summarized by the phase diagram shown in Fig.~\ref{fig:fig1}(b). Among these four phases, there is one (shown in white) that is trivial and topologically equivalent to a product state. There are three more topological phases. The first phase is topologically equivalent to the static Kitaev chain (blue). Under open boundary conditions, this phase exhibits two symmetry-protected modes at zero quasi-energy called Majorana zero modes (MZM). The remaining topological phases only occur in time-driven systems. For example, the second phase (red) exhibits a pair of Majorana $\pi$ modes (MPM) occurring at quasi-energy $\pi$ \cite{liu2013floquet}. The third phase (green) is distinct from the rest and hosts both MZM and MPM. 
Majorana modes in non-interacting systems manifest themselves by the presence of a pair of conserved boundary-localized operators $\Gamma^\omega_{s}$ that satisfy $
U^\dag_F \Gamma^\omega_s U_F= e^{-i\omega}\Gamma^\omega_s
$ \cite{jermyn2014stability,
alicea2016topological} where  index $s\in\{L,R\}$ defines right and left eigenmodes respectively, and $\omega \in\{0,\pi\}$. We will skip the frequency index $\omega$ when the context is clear. 

In the interacting case $\varphi \neq 0$, Majorana mode operators are not conserved across the spectrum, i.e. $U^\dag_F\Gamma^\omega_sU_F- e^{-i\omega} \Gamma^\omega_s= O(\tau^{-1})$. As a result, the observables associated with topological modes must decay with characteristic lifetime $\tau$. As was shown in Ref.~\cite{shtanko2020unitary}, 
if the bulk has vanishing dispersion, for small interaction angles $\varphi$ the lifetime diverges as $\tau\propto\mathcal O(\exp(c/\varphi))$, where the constant $c$ depends on the details of interaction. In practice, the lifetime may exceed dozens of Floquet cycles even if the bulk has finite dispersion and interactions are not too strong. This approximate conservation of Majorana modes leads to the persistent signal for some local observables when the rest reach infinite-temperature values. The primary goal of this work is to use this long-lived signal to restore the structure of the modes from the experiment. In this case we look for Majorana modes of the form
$
\Gamma_s = \sum_{\mu=1}^{2N}\psi^s_{\mu}\gamma_\mu,
$
where $\psi^s_\mu$ are real-valued wavefunctions. We also develop a method to distinguish trivial and topological modes.

Finally, we illustrate the exchange of Majorana modes and verify that the exchange results in the desired change of phase of the wavefunction. Conventionally, such an exchange is modeled by a slow adiabatic implementation of the unitary map $\mathcal E_{\rm ex}(\cdot) = U^\dag_{\rm ex}(\cdot)U_{\rm ex}$, where $U_{\rm ex} = \exp(-\frac \pi4 \Gamma_L\Gamma_R)$. Such a map provides $\mathcal E_{\rm ex}(\Gamma_R)  = \Gamma_L$ and $\mathcal E_{\rm  ex}(\Gamma_L)  =  -\Gamma_R$. While it is possible to carry out this procedure for one-dimensional Floquet systems \cite{bomantara2018simulation,bauer2019topologically}, it might require quantum circuits with depths beyond what is available on noisy devices. Below we show an alternative way to perform such an exchange on a noisy quantum hardware.\\

\textbf{Majorana wavefunctions}. Our first objective is to detect the presence of Majorana modes and measure the details of their structure using Fourier transformation~\cite{neil2021accurately,wikeckowski2018identification}. We assume that there are no other eigenmodes with zero or $\pi$ frequencies. In this case, we can use the asymptotic formula (see \sectionthreeexpl)
\be\label{eq:wavefnct_extract}
\begin{split}
&\psi^L_{\mu}(\omega) = F^L_\mu(\omega)/\sqrt{F^L_1(\omega)},\quad \psi^R_{\mu}(\omega) = F^R_\mu(\omega)/\sqrt{F^R_{2N}(\omega)},
\end{split}
\ee
where $\omega\in\{0,\pi\}$ is the  mode frequency, the positivity of $F^L_{1}(\omega)$  and $F^R_{2N}(\omega)$ 
is proven in \sectionthreeexpl, and
\be\label{eq:fourier_components}
F^s_\mu(\omega) = \limsum e^{i\omega n}\<\psi_0|U_F^{\dag n} \gamma^s_\mu U^{n}_F|\psi_0\>,
\ee
with $|\psi_0\> = \bigotimes_{i=1}^N|+\>_i$ being the product state of eigenstates of Pauli operator $X$ with eigenvalue one, and superscript $s\in\{L,R\}$ conforming with the representation of the Majorana operator. The order in the limit is important: one first takes the limit over the number of qubits $N$, and then the limit over the number of cycles $D$.

 In spite of the fact that the true limit cannot be reached experimentally, we measure the quantities $F^s_\mu(\omega)$ approximately using the largest available $N$ and $D$. The values of $D$ must not exceed the Majorana mode lifetime $\tau$ such that $D/\tau\ll1$. First, we initialize the qubits in the product state $|\psi_0\rangle$ and apply an $n$-cycle circuit as shown in Fig.~\ref{fig:fig1}(a) for $n=0,\dots,D-1$. For each circuit, we determine the expectation of $\gamma^{R,L}_\mu$. In the last step, we estimate the approximate value of $F_\mu(\omega)$ by summing up the results for each $n$-cycle circuit with corresponding Fourier coefficients $e^{i\omega n}/D$.

As an example, Fig.~\ref{fig:fig1}(c) shows the function $|F^L_1(\omega)|$ and its use in detecting Majorana modes and topological phases. The plots illustrate the dependence of this function on angle $\phi$ for the fixed $\theta=\pi/8$ and are similar to differential conductance spectra found in solid-state experiments \cite{beenakker2013search}. The function is equal to the topological mode density at the boundary, $F^L_1(0)=(\psi^L_{1})^2$. In particular, we observe a strong signal for $\omega=0$ in the topological phase for value $\phi = 0$, as it indicates the presence of the left MZM. Strong peaks also appear at frequencies $\omega = \pm\pi$ indicating the presence of MPM for $\phi=\pi/2$. The peaks' intensities decrease in the bulk for intermediate angles. For $\theta=\pi/8$ the boundary signal disappears at $\phi = \pi/8$ and $3\pi/8$ as the system transitions into the trivial phase.  

Next, the values of $F^s_\mu(\omega)$ for $\mu>1$ help us recover Majorana wavefunctions $\psi^s_\mu$. Plots in  Fig.~\ref{fig:fig2}(a)-(c) illustrate the normalized absolute values of wavefunctions corresponding to MZMs and MPMs in both non-interacting ($\varphi=0)$ and interacting ($\varphi=\pi/16)$ regimes. The results for the non-interacting regime are in a good agreement with the theoretical prediction. In the interacting regime, where we add an extra set of noisy two-qubit ZZ gates in each Floquet cycle, we expect to see a visibly higher level of noise in the resulting wavefunction 
 as can be seen in Fig. 2c. More data assessing the device's performance is presented \insupp.\\

\textbf{Detecting trivial modes}. Majorana modes may not be the only modes responsible for zero-frequency signals \cite{lee2014spin,kayyalha2020absence,valentini2021nontopological,yu2021non,saldana2021coulombic,wang2021spin}. In this work, we demonstrate that quantum simulators can be used to distinguish unpaired Majorana zero modes from the other topologically trivial localized excitations. Topological Majorana $\pi$ modes can be treated similarly. We use a generalized notation $\Delta_k  = \sum_{\nu}\psi_{k\nu}\gamma_\nu$ for both zero-frequency trivial and topological Majorana modes, $[\Delta_k,U_F]=0$, and $\psi_{k\nu}$ are real wavefunctions that are localized at the boundaries. In contrast to Majorana modes residing at opposite boundaries, any pair of trivial modes must always be localized near the same position.  Below we assume that the effect of disorder on the localization of the wavefunction is negligible.

We examine the two-point correlation function ($\omega=0$)
\be\label{eq:t_observable}
T_{\mu,\nu} = \limsum\<\tilde \psi_0|U_F^{\dag n} \gamma_\mu\gamma_\nu U^{n}_F|\tilde \psi_0\>,
\ee
where $|\tilde \psi_0\rangle = |\psi_{a}\>|s_2\>|s_3\>\dots|s_{N-2}\>|\psi_{a}\>
$, where $|s_i\rangle$ are random states in $Z$-basis with eigenvalues $s_i=\pm 1$, and $|\psi_{a}\rangle = \cos a|0\rangle + i\sin a|1\rangle$ for $a\in [0,\pi]$ being a phase. For simplicity, we consider  
the non-interacting case $\lambda=0$. Then the value of the correlation function for $\mu=1$ and $\nu=2$ is (see \sectionfourexpl)
\be
T_{1,2} = i\cos 2a\lim_{N\to\infty}\sum_{kk'}(\psi^2_{k1}\psi^2_{k'2}-\psi^2_{k2}\psi^2_{k'1}).
\ee
If there is only one pair of topological modes separated by the system size, then $T_{1,2} = 0$. Indeed, in this case $\sum_{kk'} \psi^2_{k1}\psi^2_{k'2}-\psi^2_{k2}\psi^2_{k'1} = (\psi^R_1)^2(\psi^L_2)^2-(\psi^R_2)^2(\psi^L_1)^2\propto O(2^{-\Theta(N)})$. A pair of trivial localized states at the left boundary, however, would result in $T_{1,2}>0$. At the same time, $T_{1,2N}$ is nonzero for both cases, while in the middle of the system, i.e. $T_{1,x} = 0$ for $x = 2cN$ and $1>c>0$. As a consequence, correlation function indicates the presence of zero-frequency modes but has a different structure for trivial and Majorana modes.

\begin{figure}[t!]
    \centering
    \includegraphics[width=\figsize\textwidth]{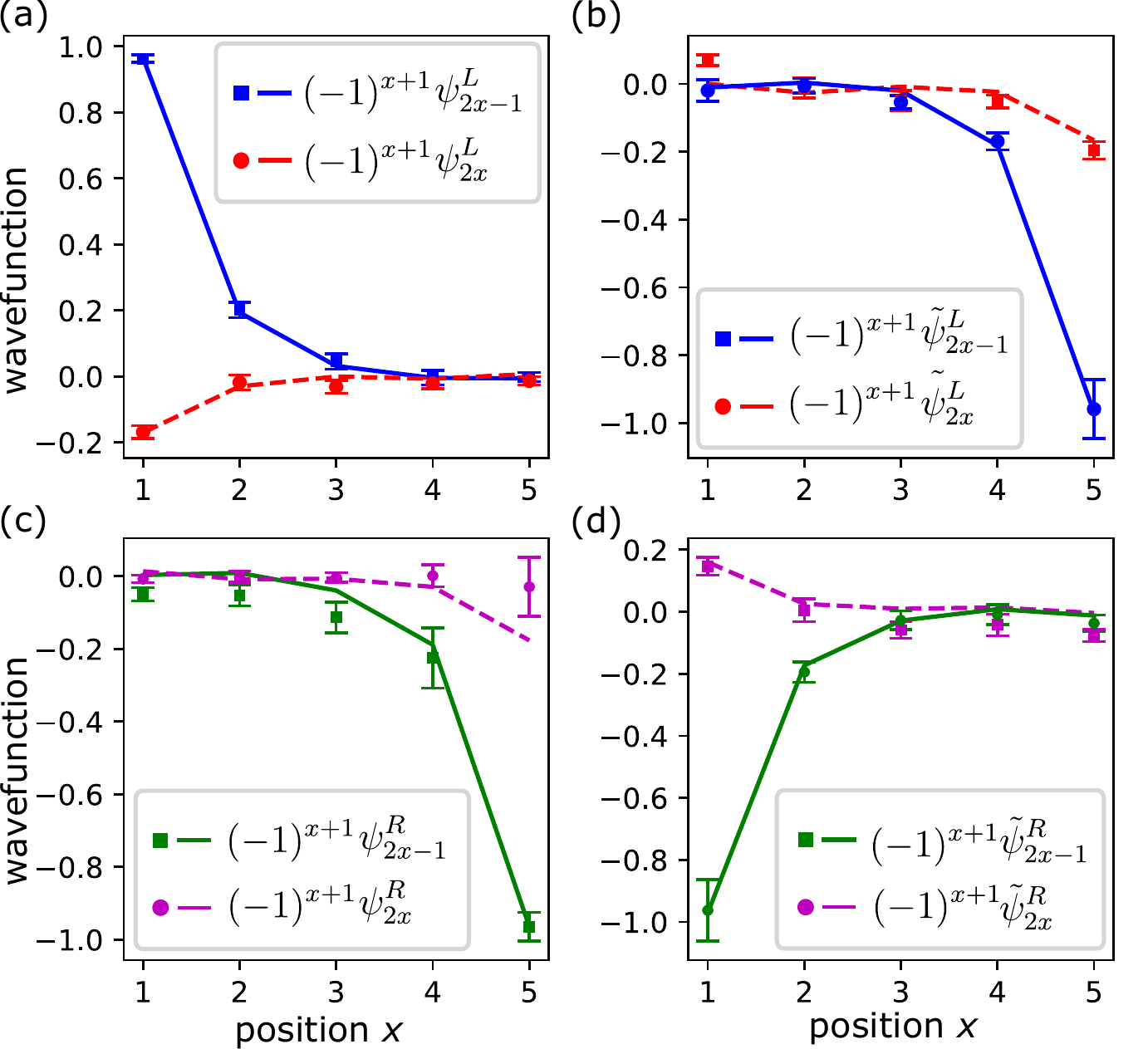}
    \caption{\textbf{Braiding.} Comparison of normalized original wavefunction in Eq.~\eqref{eq:wavefnct_extract} for the left (a) and right (c) modes and braided wavefunction in Eq.~\eqref{eq:braided_wavefunction} for the left (b) and right (d) modes with the theoretically estimated angle $\alpha_0 = 0.263127\pi$. We use the 5-qubit system on $ibm\_hanoi$ device with the parameters $\phi =\pi/16 $, $\theta=\pi/4$, and $\varphi=0$ and maximum number of cycles $D=11$, averaged over 30 experiments each with 8192 shots. Error bars are one standard of deviation. Experimental data are represented by points, whereas theoretical predictions are represented by lines. Plots illustrate that modes acquire a relative minus sign after braiding $\tilde\psi^L_\mu = \psi^R_\mu$ and $\tilde\psi^R_\mu = -\psi^L_\mu$.}
    \label{fig:fig3}
\end{figure}

In order to illustrate this method, we consider two examples of non-interacting systems, $\lambda=0$. In the first example, we use the Hamiltonian in the topological phase $(\theta=\pi/4,\phi=\pi/16)$. We compare this case to a trivial system with a slightly modified Hamiltonian. In particular, we set to zero the $XX$-term and $Z$-term for the first and the last qubits, thus decoupling them from the rest of the system (see \sectionfiveexpl). The state of the rest of the qubits
is governed by the Floquet evolution in Eq.~\eqref{eq:floquet_unitary} with parameters $(\theta=\pi/16, \phi=\pi/4)$. This modification mimics a possible error when some of the links between the qubits are dysfunctional. The modification produces two trivial full-electron modes at opposite boundaries, which is equivalent to four non-topological Majorana modes $\Delta_1=\gamma_1$, $\Delta_2=\gamma_2$, $\Delta_3=\gamma_{2N-1}$, and $\Delta_4=\gamma_{2N}$. Using only the observables in Eq.~\eqref{eq:fourier_components}, it is difficult to distinguish between these modes and topological modes. However, if we measure the sequence $|T_{1,2x}|$ for $x=1,\dots,N$ for a random configuration of the initial state, it shows an important difference. As shown in Fig.~\ref{fig:fig2}(d), the curve for trivial case is characterized by two peaks at $x=1$ and $x=N$, while topological system has only one peak around $x=N$. Thus, observation of a single peak provides reliable evidence distinguishing topological Majorana modes from the other possible trivial modes.\\

\textbf{Braiding Majorana modes}. 
Finally, we introduce a method for braiding the Majorana modes, which we call Fast Approximate Swap (FAS). Here we examine the parametrized map
\be\label{eq:braiding_map}
\mathcal E_\alpha(\cdot) := \lim_{N,D\to\infty}\frac 1D\sum_{n=0}^{D-1} U_{n\alpha}^\dag (\cdot)\,U_{n\alpha}\;,
\ee
where $U_{n\alpha} = U_F^{n} \exp(-\alpha\gamma_1\gamma_{2N})U_F^n$, and $\alpha\in[0,\pi]$ is a real parameter. This quantum channel is equivalent to selecting the unitary $U_{n\alpha}$ for $n=0,\dots D-1$ with uniform probability $1/D$.

Let us assume that the system is reflection-symmetric such that the localized modes satisfy $\psi^L_1 = \psi^R_{2N} = \xi$ and $\xi^2\geq 1/2$. Then, by setting the angle $\alpha_0 = \arcsin(1/\sqrt{2}\xi)$, the action of the map on topological Majorana operators is
\be\label{eq:exchange_map}
\mathcal E_{\alpha_0}(\Gamma_R) = p\Gamma_L, \quad \mathcal E_{\alpha_0}(\Gamma_L) = -p\Gamma_R,
\ee
where $p = \sqrt{2\xi^2-1}\leq 1$ (see \sectionfiveexpl). This procedure constitutes approximate  FAS method of braiding that aims to replace the conventional adiabatic process. This method applies in both the interacting and non-interacting regimes.

We also establish  the effect of proposed braiding map on Majorana operators in absence of localization in non-interacting limit $\lambda=0$,
\be\label{eq:map_phys_majorana}
\begin{split}
&\mathcal E_{\alpha_0}(\gamma_\mu)=p(\psi^R_\mu \Gamma_L-\psi^L_\mu \Gamma_R).
\end{split}
\ee
This allows us to detect the relative phase of Majorana fermions after braiding.
The braided mode wavefunction can be defined similarly to Eq.~\eqref{eq:wavefnct_extract} as 
\be\label{eq:braided_wavefunction}
\tilde \psi^s_\mu = \frac 1{\mathcal N} \<\psi_0|\mathcal E_{\alpha_0}(\gamma^s_\mu)|\psi_0\>,
\ee
where $\mathcal N$ is normalization coefficient.
 After braiding, assuming that the system is reflection-symmetric, we expect the braided wavefunctions to satisfy $\tilde\psi^L_\mu =\psi^R_\mu$ and $\tilde\psi^R_\mu = -\psi^L_\mu$. This behavior is illustrated in Fig.~\ref{fig:fig3}, where we compare wavefunctions in Eq.~\eqref{eq:wavefnct_extract} and Eq.~\eqref{eq:braided_wavefunction}.
 
Our braiding procedure depends on the 
parameter $\alpha_0$ that is generally unknown without prior access to the system. Although here we calculated it analytically, it may be difficult to 
find this angle theoretically for generic Hamiltonians, in which case it would be necessary to rely on experimental data. For instance, one can evaluate the angle using the measured Majorana wavefunction. In \sectionfiveexpl, we discuss an alternative method of finding the proper value of $\alpha_0$.\\

\discussion  In this work, we propose a framework for detecting, verifying, and braiding Majorana modes on near-term programmable quantum simulators by employing the Floquet dynamics. This scheme can be generalized to the continuous evolution of static Hamiltonians by replacing the discrete Fourier transformation in our work by its continuous version. It would have been possible to run our experiments on larger qubit devices. However, Majorana modes exist at the boundaries rather than the bulk, and our current experiments are sufficient to make conclusive statements about the detection and braiding of the topological Majorana modes.

The finite lifetime of the Majorana modes is attributed to natural tendency of Floquet systems to ``heat up''. Adding disorder such as randomization of phases in $Z$ gates, i.e. $\phi \to \phi+\delta_i$, where $\delta_i \in [-W,W]$, can reduce the heating because of the many-body localization (MBL) phenomenon \cite{abanin2019manybody}. However, such a simplistic scheme may require disorder values $W$ that can cause transition into the trivial phase. Avoiding phase transition would require finding good model parameters \cite{decker2020floquet} or using more sophisticated techniques \cite{shtanko2020unitary}.

This work illustrates the power of synthetic near-term qubit-based quantum computers for demonstrating and studying topological phases of electronic systems. Indeed, if we neglect noise, the observed dynamics of 
bosonic system can perfectly simulate
fermionic topological phases 
if the measurements are made 
in the non-local qubit basis. Unlike solid state devices, however, the Majorana modes in this work are subject to decoherence because noise breaks the parity symmetry protecting the topological phase in fermionic systems. This is a serious drawback for using then in topological quantum computation. Nonetheless, with improvements in coherence time, the role of noise can be sufficiently reduced as to make this type of quantum simulation useful in studying topological quantum matter.
Solid state systems such as nanowire devices \cite{lutchyn2010majorana,oreg2010helical} can be studied through continuous-time local Hamiltonian simulations. Floquet systems similar to those studied in this work, in the limit of large frequency, are equivalent to such simulations. A potential model of a nanowire could incorporate a qubit ``ladder" representing the two spin values (up and down) and local gates that account for hopping, spin-orbital coupling, and density-density interactions.

This work can be extended to regimes beyond the current classical simulation capabilities. By using devices with higher connectivity, it is possible to study generic two-dimensional materials with a broader variety of topological phases and to explore new possibilities for topological quantum computation. Further, the method studied for extracting the Majorana modes may be extended to the study of local integrals of motion in many-body localized systems, similarly to the proposal in Ref.~\cite{chandran2015constructing}.\\

\textbf{Acknowledgements}. We thank Frank Pollmann and Bela Bauer for helpful discussions. We would also like to thank Sergey Bravyi, Zlatko Minev, and Sarah Sheldon for their help in preparing this publication. The research was partly supported by the IBM Research Frontiers Institute. This research was also supported in part by the National Science Foundation under Grant No. NSF PHY-1748958.\\

\textbf{Code availability}. The code to run the experiment and access the data presented in this study is publicly available on GitHub using the link: \url{https://github.com/IBM/observation-majorana.git}

\clearpage

\pagebreak

\setcounter{page}{1}
\setcounter{equation}{0}
\setcounter{figure}{0}
\renewcommand{\theequation}{S.\arabic{equation}}
\renewcommand{\thefigure}{S\arabic{figure}}
\renewcommand*{\thepage}{S\arabic{page}}

\onecolumngrid

\begin{center}
{\large \textbf{Supplementary Information for \\``\ourtitle"}}\\
\vspace{0.25cm}

Nikhil Harle,$^{1,2}$ Oles Shtanko$^3$, and Ramis Movassagh$^{2,4}$

\vspace{0.2cm}

\textit{
$^1$\yale\\
$^2$\cambridge\\
$^3$\almaden\\
$^4$\google}
\end{center}
\vspace{1cm}

\twocolumngrid

\section{\sectionone: Hardware setting}
We replicate the Floquet dynamics in Eq.~\eqref{eq:floquet_unitary} on IBM Qiskit using the circuit in Fig.~\ref{fig:fig1}(a). We transpile the circuit on Qiskit using native gates and run it on the IBM quantum hardware. In particular, for the special value $\theta=\pi/4$ ($\varphi=\pi/4$), each $XX$-gate ($ZZ$-gate) requires a CNOT gate in combination with single qubit gates. For other non-zero angle values, two-qubit gates require two CNOTs in combination with other single-qubit gates. Therefore, to reduce the depth, part of the experiment is designed to investigate the case $\theta = \pi/4$.

\begin{figure*}[h!]
    \centering
    \includegraphics[width=0.75\textwidth]{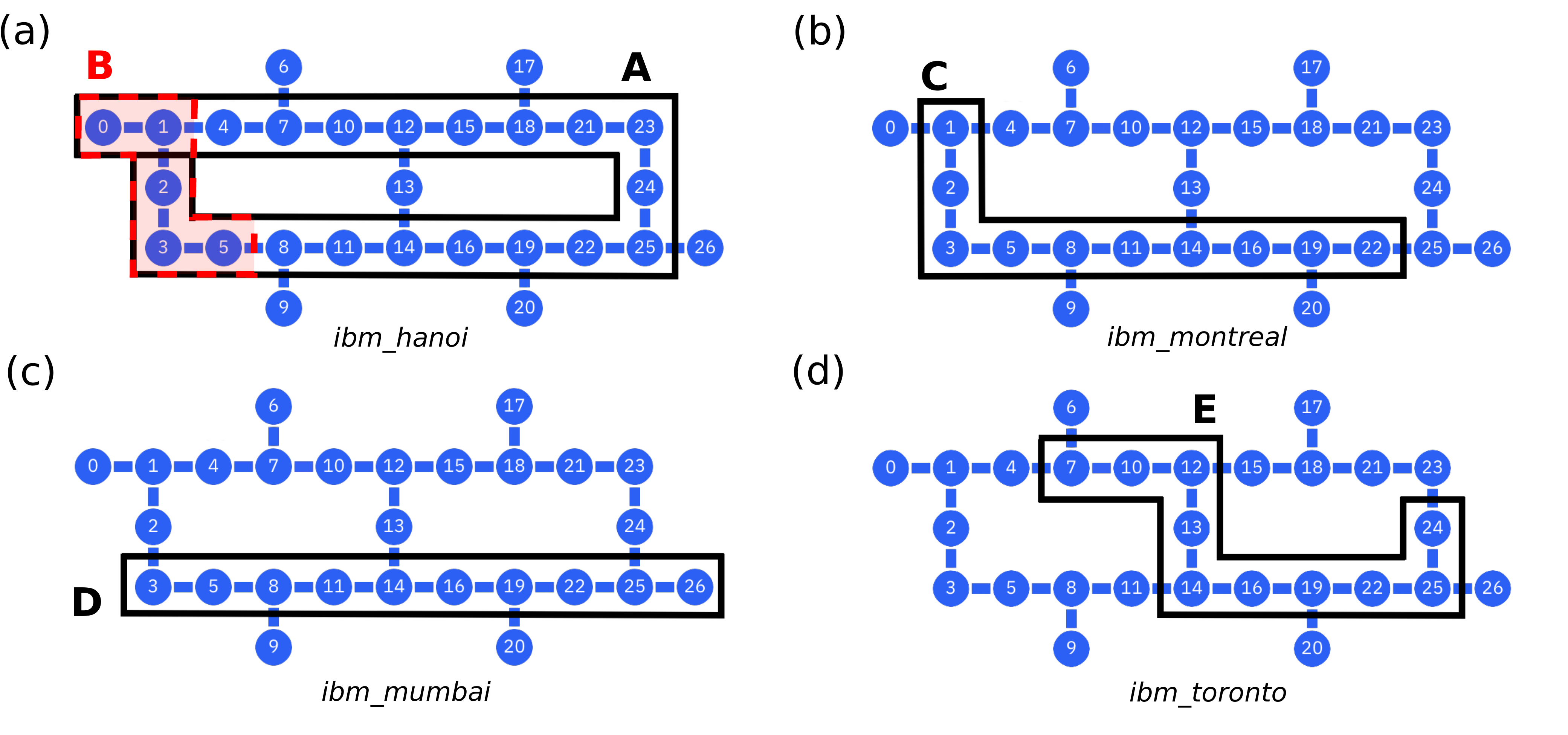}
    \caption{\textbf{Configuration of IBM hardware}. (a) Layout of the \textit{ibm\_hanoi} device. Sequence A of 21 qubits was used to generate the frequency-resolved boundary oscillations shown in Fig.~\ref{fig:fig1}(c); sequence B of 10 qubits was used to perform the braiding experiment shown in Fig.~\ref{fig:fig3}. (b) Layout for the \textit{ibm\_montreal} device. Sequence C of 10 qubits was used to reproduce the Majorana mode tomography in Fig.~\ref{fig:fig3}(a),(b) (C). Layout for the \textit{ibm\_mumbai} device. Sequence D of 10 qubits used to generate Majorana modes tomography in fig.~\ref{fig:fig2}(c). Layout for the \textit{ibm\_toronto} facility. Sequence D of 10 qubits used to generate the topological/non-topological mode separation experiment in Fig.~\ref{fig:fig2}(d).}
    \label{fig:fig1s}
\end{figure*}

\begin{figure*}[h!]
    \centering
    \includegraphics[width=1\textwidth]{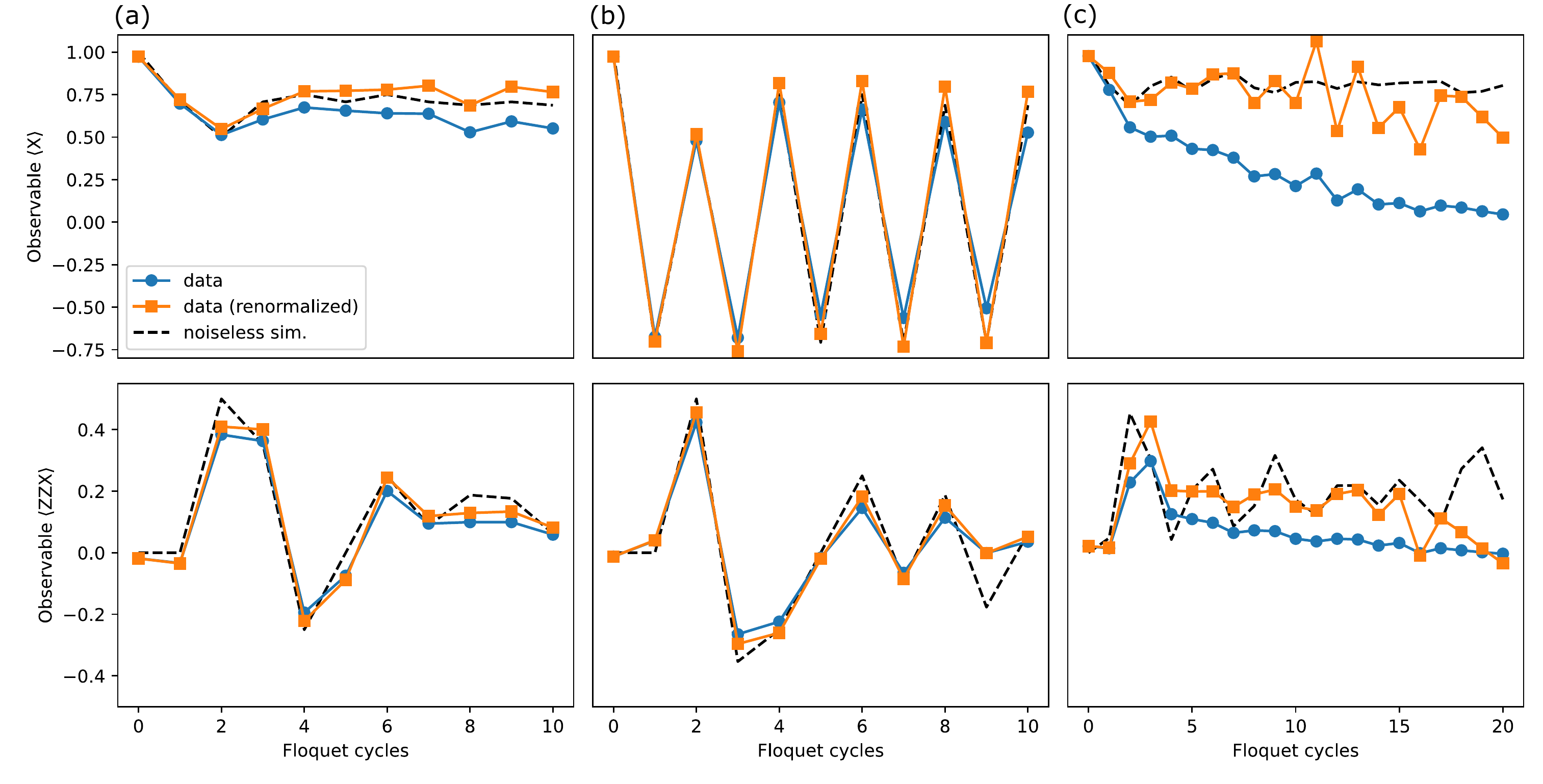}
    \caption{\textbf{Experimental data.} The figure shows the raw data for the expectations of the operators $X_0$ (top) and $Z_0Z_1X_2$ (bottom) for noiseless simulation (dashed black curve), the experimental data (circles, blue), and rescaled experimental data (squares, orange), where the rescaling takes the form $\exp(\Gamma n)$, where $n$ is the index of the Floquet cycle and $\Gamma$ is the compensated decay rate. This rescaling accounts for the decay of Majorana modes that does not contribute to the measured wavefunction \seesupp.  Panel (a) shows the results for MZM in Fig.~\ref{fig:fig2}(a) (here, we use $\Gamma = 0.0328$ for rescaling), Panel (b) shows the results for MPM in Fig.~\ref{fig:fig2}(b) (using $\Gamma = 0.0376$), and Panel (c) shows the results for MZM in the presence of $ZZ$ gates shown in Fig.~\ref{fig:fig2}(c) (using  $\Gamma = 0.120$) .}
    \label{fig:fig2s}
\end{figure*}

 For the simulation of quantum systems, Majorana operators must be encoded using qubits. We use the Jordan-Wigner transformation to implement this encoding. In particular, we define left representation as $\gamma^{L}_{2k-1} = \mathcal Z^{L}_{k}X_{k}$ and  $\gamma^{L}_{2k} = \mathcal Z^{L}_{k}Y_{k}$, where $\mathcal Z^L_k = \prod_{i=1}^{k-1} (-Z_i)$ are $Z$-string operators and $k = 1,\dots,N$. This representation is equivalent to the most common convention. Alternatively, the right representation is $\gamma^{R}_{2k-1} = \mathcal Z^{R}_{k}Y_{k}$ and $\gamma^{R}_{2k} = -\mathcal Z^{R}_{k}X_{k}$, where $\mathcal Z^R_k = \prod_{i=k+1}^N (-Z_i)$. A more traditional approach for accessing 
 these operators experimentally is to measure each qubit inside the string in the basis ($X$, $Y$, or $Z$), returning $\pm 1$ values for each qubit. In this case, we can use the product of the obtained results as the measured value. The disadvantage of 
 the conventional method is that it has lower precision due to accumulated measurement errors. Therefore, we adopt a scheme that involves only one measurement, as described below.
 
 In order to measure $\gamma^{L}_\mu$, we use the expressions
\be
\begin{split}
U^\dag_{k}X_1U_{k} = (-1)^{k-1}\mathcal Z^L_kX_{k},\qquad V^\dag_kY_1V_k = (-1)^{k-1}\mathcal Z^L_kY_{k},
\end{split}
\ee
where
\be
U_{k} := U^{YX}_1 \dots U^{YX}_{k-1} , \qquad V_k := U^{XY}_1 \dots U^{XY}_{k-1},
\ee
and can be expressed as a product of two-qubit unitaries $U^{YX}_j := \exp(-i\frac{\pi}{4} Y_j X_{j+1})$ and $U^{XY}_j := \exp(i\frac{\pi}{4} X_j Y_{j+1})$.
According to these expressions, to measure the string operator $\gamma^L_{2k-1}= \mathcal Z^L_kX_k$, we apply the gates $U^{YX}_j$ consecutively and in the reverse order for $j = k-1,\dots,1$ as shown in Fig.~\ref{fig:fig1}(a). Next, we measure the first qubit ($j=1$) in X-basis. 
 Similarly, to measure the operator $\gamma^L_{2k} = \mathcal Z^{L}_kY_k$, we perform similar gate sequence but with unitaries $U^{XY}_j$ and measure the qubit $j=1$ in $Y$-basis. Finally, the measurement of $\gamma^{R}_\mu$ can be done by mirroring the entire circuit in Fig.~\ref{fig:fig1}(a) upside down. 
 
 Next, to  evaluate $T_{1,2k}$ defined in Eq.~\eqref{eq:t_observable} we need to probe the operator $i\gamma_1\gamma_{2k} = (-1)^{k} Y_1 Z_2\dots Z_{k-1}Y_k$. In order to measure the Pauli string $Y_1 Z_2\dots Z_{k-1}Y_k$, we note that
 \be
 Y_1 Z_2\dots Z_{k-1}Y_k =(-1)^{k-1} G_1^\dag \mathcal Z^L_kY_{k}G_1
 \ee
 where $G =  HS^\dag$, where $S = {\rm diag}\{0,i\}$ is $S$-gate and $H$ the Hadamard gate. Thus, the procedure of measuring this operator is the same as $\mathcal Z^L_{k}Y_{k}$ with the difference that in the latter we apply $G_1$ before the series of $U^{XY}$ gates. 

Finally, applying the  unitary $ \exp\left(-\alpha\gamma_1\gamma_{2N}\right)$, which is necessary for generating $U_{n\alpha}$  in Eq.~\eqref{eq:braiding_map}, can be implemented in a similar manner. The implementing circuit consists of the series of gates $U^{XY}_j$ for $j=2,\dots,N-1$, followed by $\exp\left(-i\alpha Y_1Y_2\right)$, then by the series of $U^{XY\dag}_j$ in reverse order.

Results are obtained using IBM quantum hardware \cite{ibm_quantum_devices}. The experiments are performed on four different 27-qubit devices: \textit{ibm\_hanoi}, \textit{ibm\_montreal}, \textit{ibm\_mumbai}, and \textit{ibm\_toronto}. The number of qubits utilized for each experiment vary; Fig.~\ref{fig:fig1s} shows the chosen subsets.
For example, we perform experiments represented by Fig.~\ref{fig:fig2} using 10 qubits as the smallest system size that exhibits an overlap between unpaired Majorana modes which is smaller than the effect of noise. 
In contrast, braiding experiments are performed on 5 qubits as the effect of noise is stronger due for deeper circuits. The depth of the circuits are chosen to be between 11-21 cycles for most experiments.

\section{\sectiontwo: eigenmodes and their properties}

Analyzing the eigenmodes of a system can be an effective way to describe its dynamics. In this section, we introduce and analyze the properties of eigenmodes for time-periodic (Floquet) dynamics generated by the unitary operator $U_F$ in Eq.~\eqref{eq:floquet_unitary}. We prove three propositions concerning the eigenmodes of the system, which we then use in the following sections.

\begin{prop}\textbf{(Existence of eigenmodes)}\label{prop1}
For any unitary $U_F$, there exists a complete set of eigenoperators (eigenmodes) $\Delta_b$ and real eigenfrequencies $\omega_b$, $1\leq b\leq 4^N$, such that
\be\label{eqs:eigenproblem}
\begin{split}
&U_F^\dag \Delta_b U_F = e^{-i\omega_b}\Delta_b,\\
&\Tr(\Delta_b^\dag\Delta_{b'})=2^N\delta_{bb'},
\end{split}
\ee
where $N$ is the number of qubits and $\delta_{bb'}$ is the Kronecker delta.
\end{prop}

\begin{proof}
Consider $\mathcal P = \{P_\alpha: \alpha = 1,\dots,d\}$ as a complete set of $2^N\times 2^N$ Pauli basis operators, where $d=4^N$ is the dimension of the space they span. The unitary transformation of each basis operator is given by
\be\label{eqs:u_trans}
U^\dag_F P_\alpha U_F = \sum_{\beta=1}^{d} E_{\alpha\beta} P_\beta,
\ee
where
$
E_{\alpha\beta}  := 2^{-N}\Tr (U_F^\dag P_\alpha U_F P_\beta)$ are matrix elements of a real orthogonal matrix, 
\be
E_{\alpha\beta} = E^*_{\alpha\beta},\quad EE^T = E^T E = I.
\ee
To prove that $E$ is orthogonal, we use Eq.~\eqref{eqs:u_trans} to express
\be\label{eqs:doubleP}
\begin{split}
U^\dag_F P_\alpha P_{\alpha'} U_F = \sum_{\beta,\beta'=1}^{d} E_{\alpha\beta}E_{\alpha'\beta'} P_\beta P_{\beta'}.
\end{split}
\ee
Next, we take the trace for the two sides of Eq.~\eqref{eqs:doubleP} and use the orthogonality of Pauli matrices, $\Tr (P_\alpha P_\beta)= 2^N\delta_{\alpha\beta}$. Using the fact that $E_{\alpha\beta} = E^*_{\alpha\beta}$, we get that
\be
\sum_{\beta=1}^{d}E_{\alpha\beta} E^T_{\beta\alpha'} =\sum_{\beta=1}^{d} E_{\alpha'\beta}E^T_{\beta\alpha} = \delta_{\alpha\alpha'}.
\ee
Since the matrix $E$ is orthogonal, it has an orthonormal set of eigenstates and eigenvalues that lie on the unit circle in the complex plane:
\be
v^T_b E  = e^{-i\omega_b} v^T_b.
\ee
Using this set, we construct the operators
\be
\Delta_b = \sum_{\alpha=1}^{d} v_{b\alpha} P_\alpha,
\ee
where $v_{b\alpha}$ stands for the $\alpha$-th entry of $v_b$.
These operators satisfy
\be
\begin{split}
U_F^\dag \Delta_b U_F& = \sum_{\alpha=1}^{d} v_{b\alpha} U_F^\dag P_\alpha U_F = \sum_{\alpha,\beta=1}^d v_{b\alpha}E_{\alpha\beta} P_\beta\\
&= e^{-i\omega_b}\sum_{\beta=1}^dv_{b\beta}P_\beta = e^{-i\omega_b}\Delta_b,
\end{split}
\ee
which proves the first equation in Eq.~\eqref{eqs:eigenproblem}.
Furthermore, the orthogonality of the eigenstates of the operator $E$ leads to the expression
\be
\begin{split}
\Tr(\Delta^\dag_b\Delta_{b'}) &= \sum_{\alpha,\beta=1}^dv^*_{b\alpha}v_{b\beta}\Tr (P_\alpha P_\beta) \\
&= 2^N\sum_{\alpha=1}^dv^*_{b\alpha}v_{b'\alpha} = 2^N\delta_{bb'},
\end{split}
\ee
proving the second equation in Eq.~\eqref{eqs:eigenproblem} and thus completing our proof.\end{proof}

The spectrum of the system can contain degeneracies. 
If we define $B$ as the full set of eigenmodes, then $B_\omega = \{b_1,\dots, b_m\}$ is defined as the subset of eigenmodes with the same frequency $\omega_{b_k} = \omega$ for $1\leq k\leq m$ in the limit $N\to\infty$. We define $\overline{B_\omega} = B\setminus B_\omega$ as the set of eigenmodes with frequencies different from $\omega$.

\begin{prop} For the dynamics in Eq.~\eqref{eq:floquet_unitary}, non-trivial eigenmodes $\Delta_b\neq I$ and Majorana operators satisfy
\be\label{eqs:delta_gamma_connection}
\Delta_b = \sum_{\mu=1}^{2N} w_{\mu b}\gamma_\mu + C_\perp, \quad \gamma_\mu = \sum_{b\in B} w^*_{\mu b}\Delta_b,
\ee
where $w_{\mu b}$ are complex-valued coefficients and $C_\perp$ is an operator that satisfies $\Tr(\gamma_\mu C_\perp)=0$ for all $\mu$.
\end{prop}

\begin{proof}
First, we prove the second equation.
Since $\Delta_b$ form a complete basis, we have
\be\label{eqs:decomposition}
O = \frac 1{2^N} \sum_{b\in B} \Tr(O \Delta^\dag_b)\Delta_b.
\ee
Then, the second part of Eq.~\eqref{eqs:delta_gamma_connection} follows from this expression by setting $w^*_{\mu b} := 2^{-N}\Tr(\gamma_\mu\Delta^\dag_b)$.
Any many-body operator can be written as a decomposition of Majorana operators,
\be
\Delta_b = \sum_{\mu=1}^{2N} C^{(1)}_{b,\mu}\gamma_\mu+ \sum_{\mu\nu=1}^{2N} C^{(2)}_{b,\mu\nu}\gamma_\mu\gamma_\nu+\dots\;.
\ee
Multiplying both sides by the operator $\gamma_\mu$ and taking the trace, we get
\be
\Tr(\Delta_b\gamma_\mu) = 2^N C^{(1)}_{b,\mu}.
\ee
That is, we derive the expression
\be
C^{(1)}_{b,\mu} = w_{\mu b}. 
\ee
which leads to Eq.~\eqref{eqs:delta_gamma_connection} and completes our proof.
\end{proof}

As a useful tool for analytical derivations, we introduce the Fourier channel
\be\label{eqs:fourier_channel}
\mathcal F_\omega(\cdot) := \limsum e^{i\omega n} U_F^{\dag n} (\cdot)U^n_F,
\ee
where we first take $N$ and then $D$ to infinity. The action of this channel is expressed in terms of the eigenmodes of the system by the following Proposition.

\begin{prop}
 The action of the Fourier map in Eq.~\eqref{eqs:fourier_channel} can be expressed as
\be\label{eqs:fourier_rep}
\mathcal F_{\omega}(O) =\lim_{N\to\infty} \frac 1{2^N}\sum_{b\in B_\omega}\Tr(O\Delta^\dag_b)\Delta_b.
\ee
\end{prop}

\begin{proof}
Using the decomposition in Eq.~\eqref{eqs:decomposition}, we have
\be
\mathcal F_{\omega}(O) = \limsum e^{i\omega n}\sum_{b\in B} \frac 1{2^N}\Tr(O\Delta^\dag_b)U_F^{\dag n} \Delta_b U^{n}_F.
\ee
Now, using Proposition~\ref{prop1}, we have
\be
\begin{split}
\mathcal F_{\omega}(O) = \lim_{N,D\to\infty}\frac 1{2^ND}\sum_{n=0}^{D-1}\sum_{b\in B} e^{i(\omega-\omega_b) n}\Tr(O\Delta^\dag_b)\Delta_b.
 \end{split}
\ee
Using the property
\be
\lim_{D\to\infty}\frac 1D \sum_{n=0}^{D-1} e^{i\omega n} = \delta_{\omega,0},
\ee
we get the expression
\be
\mathcal F_{\omega}(O) =  \lim_{N\to\infty}\frac1{2^N}\sum_{b\in B} \delta_{\omega-\omega_b,0}\Tr(O\Delta^\dag_b)\Delta_b.
\ee
Next, using the definition of $B_\omega$, we obtain the statement of this Proposition.
\end{proof}

If the system is non-interacting, i.e. $\lambda=0$, the eigenmodes can be efficiently expressed in the free fermionic representation. The action of the Floquet unitary on the Majorana operators is then the linear map
\be\label{eqs:linear_transformation}
U^\dag_F\gamma_\mu U_F = \sum_{\nu=1}^{2N} u_{\mu\nu}\gamma_\nu,
\ee
where $u$ is a $2N\times 2N$ unitary matrix
\be\label{eqs:free_fermion_unitary}
u = \exp(4\theta h_{xx})\exp(4\phi h_z),
\ee
and $h_z$ and $h_{xx}$ are
\be
\begin{split}
& h_{z} = -\frac 12\sum_{k=0}^{N-1}\Bigl(|2k+1\>\<2k+2|-{\rm h.c.}\Bigl),\\
& h_{xx} = \frac 12\sum_{k=0}^{N-2}\Bigl(|2k+2\>\<2k+3|-{\rm h.c.}\Bigl).
\end{split}
\ee
Using the solution of the eigenproblem
\be\label{eqs:single_fermion_eigenproblem}
\psi^T_k u = e^{-i\omega^0_k} \psi^T_k,
\ee
we define the set of single-fermion modes
\be\label{eq:single_fermion_modes}
\Delta^0_k = \sum_{\mu=1}^{2N} \psi_{k\mu}\gamma_\mu, \quad \{\Delta^{0\dag}_k,\Delta^0_{k'}\} = 2\delta_{kk'}.
\ee
Since the unitary matrix in Eq.~\eqref{eqs:free_fermion_unitary} is real-valued, the complex conjugate for both sides in Eq.~\eqref{eqs:single_fermion_eigenproblem} gives us another, opposite-frequency solution
\be
(\psi^*_k)^Tu = e^{i\omega^0_k} (\psi^*_k)^T.
\ee
This conclusion proves that the spectrum of the problem is zero-symmetric, i.e. for each mode $k$ there exists an orthogonal mode $\sigma(k)$ such that $\psi_{\sigma(k)} = \psi_k^*$ and $\omega_{\sigma(k)}^0 = -\omega_k^0$. This mathematical property is a reflection of the physical particle-hole symmetry.

For the non-interacting case, each many-body eigenmode in Eq.~\eqref{eqs:eigenproblem} and its frequency can be expressed as
\be\label{eqs:composite_mode}
\Delta_b = \Delta^0_{m_1}\dots \Delta^0_{m_K}, \qquad
\omega_b = \sum_{i=1}^K\omega^0_{m_i}.
\ee
Each mode appears only once in the product, i.e. $m_i\neq m_j$, and the number of single-particle modes involved is $K\in \{1,\dots,2N\}$. A crucial advantage of the free fermion representation is its efficiency in evaluating the eigenmodes on a classical computer by diagonalizing the unitary operator in Eq.~\eqref{eqs:free_fermion_unitary}.

Finally, we define the sets of approximate modes.

\begin{dfn} 
An orthonormal operator set $\{\Delta'_b\}$ is called an $\epsilon$-approximate set of eigenmodes if it satisfies
\be\label{eqs:approx_eigenproblem}
\begin{split}
&U_F^\dag \Delta'_b U_F = e^{-i\omega_b}\Delta'_b+\delta O, \quad \|\delta O\|\leq \epsilon,\\
\end{split}
\ee
where $\|\cdot\|$ is the operator norm.
\end{dfn}

Approximate modes behave similarly to regular eigenmodes up to times $\tau \propto 1/\epsilon$.

 Let us consider the realistic case $D<\infty$ and in the presence of noise. In the experiment, we implement the map
\be
\tilde{\mathcal F}_{D,\omega}(\cdot) := \frac 1D\sum_{m=0}^{D-1} e^{i\omega n} \mathcal E^n_F (\cdot),
\ee
where $\mathcal E_F$ is a general map that includes unitary evolution and the effect of noise. The effect of this map can be approximated by
\be\label{eqs:approx_model}
\mathcal E_F (\Delta_b) \approx e^{i\omega_b -\Gamma_b }\Delta_b
\ee
where $\Gamma_b$ is the decay rate that includes both the effect of the noise and the finite lifetime of the approximate integral of motion $\Delta_b$. This effect is visible in Fig.~\ref{fig:fig2s}, especially in the presence of $ZZ$-gates. Then, the effect of noise and finite depth can be understood as
\be\label{eqs:realistic_map}
\tilde{\mathcal F}_{D,\omega}(O) \approx  \sum_{b\in B}f_D(\omega-\omega_b, \Gamma_b)\Tr(O\Delta^\dag_b)\Delta_b.
\ee
where the coefficients are
\be
f_D(\omega,\Gamma) = \frac 1D \frac{1-e^{i\omega D-\Gamma D}}{1-e^{i\omega-\Gamma}}.
\ee
This function generates the peak broadening and effective attenuation factor as $|f_D(\omega, \Gamma)|\leq 1$.

\section{\sectionthree: Majorana eigenmodes} 

The topological Majorana modes $\Gamma_s$, $s\in\{L,R\}$, satisfy
\be\label{eqs:integrals of motion}
U^\dag_F \Gamma_s U_F= e^{-i\omega_M}\Gamma_s+O(1/\tau),
\ee
where $\omega_M$ are the Majorana mode frequencies, $\omega_M \in\{0,\pi\}$ in the limit $N\to\infty$, and $\tau$ is the lifetime. The Majorana modes are commonly described as ``strong'' if $\tau\to\infty$, while the case $\tau<\infty$ corresponds to ``weak'' modes.

For the topological Majorana modes, both strong and weak, we use the following form
\be\label{eqs:Gtog}
\Gamma_s = \sum_{\mu=1}^{2N} \psi^s_\mu \gamma_\mu,
\ee
 where $\gamma_\mu$ are physical Majorana operators and $\psi^s_\mu$ are real Majorana wave functions. In the non-interacting topological phase, $\lambda=0$, there always exists a pair of solutions for the unitary in Eq.~\eqref{eqs:free_fermion_unitary} which satisfies
 \be\label{eq:eigenproblem}
\sum_{\mu=1}^{2N}u_{\mu\nu}\psi^s_\mu = \pm\psi^s_\nu
\ee
corresponding to MZMs (plus sign) and MPMs (minus sign). Therefore, non-interacting regimes are always characterized by \textit{strong} Majorana modes. On the other hand, in the interacting regime, i.e. $|\lambda|>0$, the existence of such operators is not guaranteed. However, for certain values of the angle $(\theta,\phi)$ such modes are weak eigenmodes characterized by large $\tau\propto \exp(c/\lambda)\gg1$ \cite{shtanko2020unitary}. Even in the ideal case where we ignore the noise in the actual experiment, as will become clear in Eq.~\eqref{eqs:majorana_fourier_e}, the lifetime $\tau$ limits the maximum depth $D$ such that $D/\tau\to0$. In practice, this connection means that the Fourier components evaluated in Eq.~\eqref{eq:fourier_components} must be compared with the experimental results obtained for $\tau\gg D\gg 1$.

For both strong and weak modes, the wavefunctions $\psi^s_\mu$ must be properly normalized, as follows from the condition $\Gamma_s^2 = 1$. Indeed,
\be
\Gamma_s^2 = \sum_{\nu,\nu'=1}^{2N} \psi^s_{\nu}\psi^s_{\nu'}\gamma_\nu\gamma_{\nu'} = \sum_{\nu=1}^{2N} (\psi^{s}_{\nu})^2 = 1,
\ee
where we used the statistics of the Majorana operators, $\{\gamma_\mu,\gamma_\nu\} = 2\delta_{\mu\nu}$. 

Next, we analyze how to detect Majorana modes using the Fourier components in Eq.~\eqref{eq:fourier_components} generated from the experiment. To do this, we express these components using the Fourier map in Eq.~\eqref{eqs:fourier_channel} as
\be\label{eqs:fourier_majorana_expression}
\begin{split}
F^s_\mu(0) := \limsum \<\psi_0|&U_F^{\dag n}
\gamma_\mu U^{n}_F|\psi_0\>\\
& =\<\psi_0|\mathcal F_0(\gamma_\mu)|\psi_0\>.
\end{split}
\ee
In the following section we will show how to distinguish true MZM modes from other trivial edge oscillations. For now, let us focus on zero-frequency MZMs and assume that they are the only pair of zero-frequency modes that have a non-zero overlap with the Majorana operators $\gamma_\mu$.  Then, from Proposition 3, this action is given by
\be
\begin{split}
\mathcal F_{0}(\gamma_\mu) &= \sum_{s\in\{L,R\}}\psi^s_\mu\limsum  U_F^{\dag n}
\Gamma_s U^{n}_F\\
&\quad +\limsum\frac{1}{2^N}\sum_{b\in \overline{B_0}}w^*_{b\mu}e^{-i\omega_b n}\Delta_b.
\end{split}
\ee
where $\overline{B_0}$ is the complementary to zero frequency subspace $B_0$.
Under the limit over $D$ the last term of this expression vanishes. At the same time, the first term can be simplified with Eq.~\eqref{eqs:integrals of motion} as
\be\label{eqs:majorana_fourier_e}
\limsum  U_F^{\dag n}
\Gamma_s U^{n}_F = \limsum  (\Gamma_s+O(n/\tau))
\ee
Assuming that $\lim_{D\to\infty} D/\tau = 0$, we get
\be\label{eqs:fourier}
\begin{split}
\mathcal F_{0}(\gamma_\mu)=\lim_{N\to\infty} \sum_{s\in\{L,R\}}\psi^s_\mu\Gamma_s\equiv \lim_{N\to\infty}(\psi^L_\mu\Gamma_L+\psi^R_\mu\Gamma_R),
\end{split}
\ee
Using Eqs.~\eqref{eqs:fourier} and \eqref{eqs:Gtog}, we transform Eq.~\eqref{eqs:fourier_majorana_expression} into
\be
F^\alpha_\mu(0) = \lim_{N\to\infty}\sum_{s\in \{L,R\}}\sum_{\mu'=1}^{2N} \psi^s_{\mu} \psi^s_{ \mu'}\<\psi_0|\gamma^\alpha_{\mu'}|\psi_0\>.
\ee
where we restored the Pauli representation index for Majorana operators. 
Using $\<\psi_0|\gamma^L_\mu|\psi_0\> = \delta_{\mu,1}$ and $\<\psi_0|\gamma^R_\mu|\psi_0\> = \delta_{\mu,2N}$,
 we rewrite
\be
\begin{split}
&F_\mu^L(0) = \lim_{N\to\infty}(\psi^L_{\mu}\psi^L_1+\psi^R_{\mu}\psi^R_1),\\
&F^R_\mu(0) = \lim_{N\to\infty}(\psi^L_{\mu}\psi^L_{2N}+\psi^R_{\mu}\psi^R_{2N}).
\end{split}
\ee
Taking into account that the eigenmodes are exponentially suppressed away from the respective boundaries of the system, i.e. $\psi^R_1 \sim \psi^L_{2N} \sim 2^{-\Theta(N)}$, we obtain the expressions
\be
\begin{split}
&F_\mu^L(0) = \lim_{N\to\infty}\psi^L_{\mu}\psi^L_1,\\
&F^R_\mu(0) = \lim_{N\to\infty}\psi^R_{\mu}\psi^R_{2N},
\end{split}
\ee
from which we derive
\be\label{eqs:wavefunctions_F}
\begin{split}
&\psi^L_{\mu} = F^L_\mu(0)/\sqrt{F^L_1(0)},\\
&\psi^R_{\mu} = F^R_\mu(0)/\sqrt{F^R_{2N}(0)}.
\end{split}
\ee
The evaluation can be done in a similar way for MPM.
The final expression is the same as Eq.~\eqref{eq:wavefnct_extract} in the main text.

In the case of finite circuit depth and in the presence of noise, we use Eq.~\eqref{eqs:realistic_map} to modify Eq.~\eqref{eqs:wavefunctions_F} as
\be\label{eqs:wvafunction_normalization}
\psi^{L,R}_{\mu} \to \sqrt{|f(\omega,\Gamma)|}\,\psi^{L,R}_{\mu},
\ee
where $\omega\sim 2^{-O(N)}$ is the frequency of the Majorana mode for a system of finite size and $\Gamma$ is the effective decay rate. Since $|f(\omega,\Gamma)|\leq1$, the approximation given by Eq.~\eqref{eqs:approx_model} leads to a uniform damping of the wavefunction as seen in Eq.~\eqref{eqs:wvafunction_normalization}. Normalizing the wavefunction eliminates this effect. However, more general noise models would result in more complex effects.

\section{\sectionfour: Two-point correlation function}

In this section we will assume a non-interacting scenario where $\lambda=0$. For simplicity, we will focus on zero-frequency modes, although $\pi$-frequency modes can be treated similarly. First, we will establish a connection between the two-point function in Eq.~\eqref{eq:t_observable} and an expectation of the Fourier map given by
\be
T_{\mu\nu} = \<\tilde\psi_0|\mathcal F_0(\gamma_\mu\gamma_\nu)|\tilde\psi_0\>
,\ee
where we have chosen the initial product state in the form of the product state $|\tilde\psi_0\> = |\psi_a\>|s_2\>\dots|s_{N-1}\>|\psi_a\>$, where $|\psi_a\> = \cos a|0\>+i\sin a|1\>$, and $|s_i\>$ are arbitrary states in the $Z$ basis, $s_i\in\{0,1\}$.
Next, we express the Majorana operators using free fermion eigenmodes as
\be
\gamma_\mu = \sum_{k=1}^{2N}\psi^*_{k\mu}\Delta^0_k,
\ee
where the wavefunctions $\psi_{k\mu}$ are defined in Eq.~\eqref{eqs:single_fermion_eigenproblem}.
This decomposition allows us to write
\be\label{eqs:T_operator}
\begin{split}
\mathcal F_{0}&(\gamma_\mu\gamma_\nu)=\\
&= \limsum\sum_{k,k'=1}^{2N}\psi^*_{k\mu}\psi^*_{k'\nu} e^{-i(\omega^0_{k}+\omega^0_{k'}) n}\Delta^0_k\Delta^0_{k'}\\
& = \lim_{N\to\infty}(A_{\mu\nu}+B_{\mu\nu}),
\end{split}
\ee
where the variable $A_{\mu\nu}$ denotes the contribution from the set of zero-frequency modes ($B_0$), while $B_{\mu\nu}$ represents the contribution from pairs of modes with opposite frequencies from the complementary set $\overline{B_0}$,
\be
\begin{split}
&A_{\mu\nu} = \sum_{k,k'\in B_0}\psi_{k\mu}\psi_{k'\nu}\Delta^0_k\Delta^0_{k'},\\
&B_{\mu\nu} = \sum_{k\in \overline{B_0}}\psi^*_{k\mu}\psi^*_{\sigma(k)\nu}\Delta^0_k\Delta^0_{\sigma(k)}\\
& \qquad \qquad =  \sum_{k\in\overline{B_0}}\psi^*_{k\mu}\psi_{k\nu}\Delta^0_k\Delta^{0\dag}_k
\end{split}
\ee
The remaining terms vanish in the limit $D\to\infty$.
Here $\sigma(k)$ represents the opposite frequency mode with respect to mode $k$. We also take advantage of the fact that we can always choose the zero-frequency modes $k\in B_0$ to be real-valued, $\psi_{k\mu} = \psi^*_{k\mu}$, while the remaining modes satisfy the relations $\psi_{\sigma(k)\mu} = \psi^*_{k\mu}$, $\omega_{\sigma(k)}^0 = -\omega_k^0$, and $\Delta^0_{\sigma(k)} = \Delta^{0\dag}_k$. 

Using these notations, the target two-point correlation function can be expressed as
\be
T_{\mu\nu} = \lim_{N\to\infty}\Bigl(\<\tilde\psi_0|A_{\mu\nu}|\tilde\psi_0\>+\<\tilde\psi_0|B_{\mu\nu}|\tilde\psi_0\>\Bigl).
\ee
Next, we use Eq.~\eqref{eqs:delta_gamma_connection} to express the expectation value
\be
\begin{split}
&\<\tilde\psi_0|\Delta^0_k\Delta^0_{k'}|\tilde\psi_0\> = \sum_{\mu,\mu'=1}^{2N}\psi_{k\mu}\psi_{k'\mu'}\<\tilde\psi_0|\gamma_\mu\gamma_{\mu'}|\tilde\psi_0\>,\\
&\<\tilde\psi_0|\Delta^0_k\Delta^{0\dag}_{k'}|\tilde\psi_0\> = \sum_{\mu,\mu'=1}^{2N}\psi_{k\mu}\psi^*_{k'\mu'}\<\tilde\psi_0|\gamma_\mu\gamma_{\mu'}|\tilde\psi_0\>.
\end{split}
\ee
The expected values on the right-hand side of these equations can be expressed in terms of the chosen product state $|\tilde \psi_0\>$ as
\be
\begin{split}
\<\tilde\psi_0|\gamma_\mu\gamma_{\mu'}|\tilde\psi_0\> = \delta_{\mu\mu'}&+iC_1(\delta_{\mu1}\delta_{\mu'2N}-\delta_{\mu2N}\delta_{\mu'1})\\
&+iC_2(\delta_{\mu1}\delta_{\mu'2}-\delta_{\mu2}\delta_{\mu'1})\\
&+iC_2(\delta_{\mu2N-1}\delta_{\mu'2N}-\delta_{\mu2N}\delta_{\mu'2N-1}),
\end{split}
\ee
where we used the notations
\be
\begin{split}
&C_1 = \<\psi_a|Y|\psi_a\>^2\prod_{i=2}^{N-1}\<s_i|(-Z_i)|s_i\> = \, (-1)^{N+S}\sin^2 (2a),\\
&C_2 = \<\psi_a|Z|\psi_a\> = \cos 2a.
\end{split}
\ee
where $S = \sum_{i=2}^{N-1}s_i$.
We use these expressions to obtain
\be
\begin{split}
\<\tilde\psi_0|\Delta^0_k\Delta^0_{k'}|\tilde\psi_0\> & = \sum_{\mu=1}^{2N}\psi_{k\mu}\psi_{k'\mu}+iC_1(\psi_{k1}\psi_{k'2N}-\psi_{k2N}\psi_{k'1})\\
&\qquad +iC_2(\psi_{k1}\psi_{k'2}-\psi_{k1}\psi_{k'2}\\
&\qquad\qquad +\psi_{k2N-1}\psi_{k'2N}-\psi_{k2N}\psi_{k'2N-1}).
\end{split}
\ee
and
\be
\begin{split}
\<\tilde\psi_0|\Delta^0_k\Delta^{0\dag}_{k'}|\tilde\psi_0\> & = \sum_{\mu=1}^{2N}\psi_{k\mu}\psi^*_{k'\mu}+iC_1(\psi_{k1}\psi^*_{k'2N}-\psi_{k2N}\psi^*_{k'1})\\
&\qquad +iC_2(\psi_{k1}\psi^*_{k'2}-\psi_{k1}\psi^*_{k'2}\\
&\qquad\qquad +\psi_{k2N-1}\psi^*_{k'2N}-\psi_{k2N}\psi^*_{k'2N-1}).
\end{split}
\ee
For real-valued zero frequency modes, we use the orthogonality condition
\be
 \forall k\in B_0:\qquad \sum_{\mu=1}^{2N}\psi_{k\mu}\psi_{k'\mu} = \delta_{kk'}.
\ee
As a result, the expression for the contribution of the zero frequency modes to the two-point correlation function can be written as
\be\label{eqs:a_part}
\begin{split}
\<\tilde\psi_0|&A_{\mu\nu}|\tilde\psi_0\> =  \sum_{k,k'\in B_0}\psi_{k\mu}\psi_{k'\nu}\<\tilde\psi_0|\Delta^0_k\Delta^0_{k'}|\tilde\psi_0\> \\
 = &  \sum_{k\in B_0} \psi_{k\mu}\psi_{k\nu}\\
& + iC_1\Bigl(\sum_{k\in B_0}\psi_{k\mu}\psi_{k1}\sum_{k'\in B_0}\psi_{k'\nu}\psi_{k'2N}-(\mu\leftrightarrow\nu)\Bigl)\\
& + iC_2\Bigl(\sum_{k\in B_0}\psi_{k\mu}\psi_{k1}\sum_{k'\in B_0}\psi_{k'\nu}\psi_{k'2}-(\mu\leftrightarrow\nu)\Bigl)\\
& + iC_2\Bigl(\sum_{k\in B_0}\psi_{k\mu}\psi_{k2N-1}\sum_{k'\in B_0}\psi_{k'\nu}\psi_{k'2N}-(\mu\leftrightarrow\nu)\Bigl),
\end{split}
\ee
where $(\mu\leftrightarrow\nu)$ is the same term as before with the indices $\mu$ and $\nu$ swapped.

Similarly, the second part has the form
\be\label{eqs:b_part}
\begin{split}
\<\tilde\psi_0|B_{\mu\nu}&|\tilde\psi_0\> = \sum_{k\in \overline{B_0}} \psi^*_{k\mu}\psi_{k\nu}\\
&+\Bigl( iC_1\sum_{k\in \overline{B_0}}\psi^*_{k\mu}\psi_{k1}\psi_{k\nu}\psi^*_{k2N}\\
& + iC_2\sum_{k\in \overline{B_0}}\psi^*_{k\mu}\psi_{k1}\psi_{k\nu}\psi^*_{k2}\\
& + iC_2\sum_{k\in \overline{B_0}}\psi^*_{k\mu}\psi_{k2N-1}\psi_{k\nu}\psi^*_{k2N}-(\mu\leftrightarrow\nu)\Bigl).
\end{split}
\ee
Since we assume that the bulk of the system is delocalized, the wavefunctions of single-fermion modes with non-zero frequency must satisfy $\psi_{k\mu}\propto O(N^{-1/2})$. This means that all terms in Eq.~\eqref{eqs:b_part} except the first have $O(N^{-1})$ scaling and therefore vanish in the limit $N\to\infty$. Therefore, combining the contributions in Eqs.~\eqref{eqs:a_part} and \eqref{eqs:b_part} we get
\be
\begin{split}
T_{\mu,\nu} =  \delta_{\mu\nu} & + \lim_{N\to\infty} \Bigl[iC_1\sum_{k\in B_0}\psi_{k\mu}\psi_{k1}\sum_{k'\in B_0}\psi_{k'\nu}\psi_{k'2N} \\
& + iC_2\sum_{k\in B_0}\psi_{k\mu}\psi_{k1}\sum_{k'\in B_0}\psi_{k'\nu}\psi_{k'2}
\\
& + iC_2\sum_{k\in B_0}\psi_{k\mu}\psi_{k2N-1}\sum_{k'\in B_0}\psi_{k'\nu}\psi_{k'2N}\\
& -(\mu\leftrightarrow\nu) \Bigl].\end{split}
\ee
Next, we analyze the behavior of this function for a few characteristic values of $\mu$ and $\nu$. First we consider the combination $\mu = 1$ and $\nu = 2$ and assume that the wavefunctions $\psi_{k\mu}$ for zero modes $k\in B_0$ are strongly localized, i.e. $\psi_{k\mu}\psi_{k\nu} \leq 2^{-c|\mu-\nu|}$, $c>0$. Then the correlation function takes the form
\be
T_{1,2} = iC_2\sum_{k\neq k'\in B_0}(\psi^2_{k1}\psi^2_{k'2}-\psi^2_{k2}\psi^2_{k'1}).
\ee
In particular, if the only modes are Majorana modes, then $\psi^L_{1}\psi^R_{2} \sim \psi^L_1\psi^R_2\sim O(2^{-\Theta(N)})$, and this quantity vanishes in the thermodynamic limit. In contrast, if there is a localized state at the boundary of the system, there exists a pair $b\neq b'$ for which $T_{1,2}= O(1)$.

At the same time, the correlation function for $\mu = 1$ and $\nu = 2N$ has the form
\be
T_{1,2N} = iC_1\sum_{b,b'\in B_0}(\psi^2_{k1}\psi^2_{k'2N}-\psi^2_{k2N}\psi^2_{k'1}).
\ee
This expression does not vanish for both Majorana and trivial modes. For other points $\mu$ and $\nu$ far from the boundaries, the two-point correlation function vanishes.

To model the trivial system with approximate zero-energy localized boundary eigenmodes, we consider the Hamiltonian
\be
H(t) = \sum_{j=1}^{N-1}J_j(t)\,X_j X_{j+1}+\sum_{j=1}^{N}h_j(t) Z_j ,
\ee
with the same protocol. The couplings and fields for the bulk qubits are the same, $J_jT = \pi/16$ for $j\neq 1,N-1$ and $h_jT = \pi/4$ for $j\neq 1,N$. At the same time, we set $J_1 = J_{N-1} = h_1 = h_{N} = 0$. Due to decoupling of the boundary qubits, Majorana operators corresponding to these qubits (i.e. $\gamma_1$, $\gamma_2$, $\gamma_{2N-1}$, and $\gamma_{2N}$) are integrals of motion.

\section{\sectionfive: Braiding map}

In this section, we provide a rigorous proof of the properties of the braiding map in Eq.~\eqref{eq:braiding_map} in the main text. First, we formulate Lemma 1, which establishes the action of the map on MZM operators.
\begin{lem} 
Suppose Eq.~\eqref{eqs:fourier} holds and $\psi^L_1 = \psi^R_{2N} = \xi$, $\xi^2\geq 1/2$. Then for the angle $\alpha_0 = \arcsin(1/\sqrt{2}\xi)$ the action of the map in Eq.~\eqref{eq:braiding_map} on MZM operators is
\be
\mathcal E_{\alpha_0}(\Gamma_R) = p\Gamma_L, \quad \mathcal E_{\alpha_0}(\Gamma_L) = -p\Gamma_R,\ee
where $p = \sqrt{2\xi^2-1}$.
\end{lem}

Lemma 1 leads to Eq.~\eqref{eq:exchange_map} in the main text. This lemma applies only to systems where MZMs are the only zero-frequency modes overlapping with single-fermion operators, as manifested by Eq.~\eqref{eqs:fourier}. At the same time, it applies to a generic setting including the interacting case $|\lambda|>0$.

Next, we formulate Lemma 2, which gives the expression for the map action on Majorana operators.

\begin{lem}
Under conditions of Lemma 1, the map in Eq.~\eqref{eq:braiding_map} satisfies
\be
\mathcal E_{\alpha_0}(\gamma_\mu)=p(\psi^R_\mu \Gamma_L-\psi^L_\mu \Gamma_R)+\delta C,
\ee
where the norm of the correction operator is bounded as
\be\label{correction}
\begin{split}
&\|\delta C\| \leq \frac 1{\xi^2}\max_\mu\sqrt{\sum_\nu\kappa_{\mu\nu}^2}, \\
&\kappa_{\mu\nu} = \lim_{N\to\infty}\sum_{k\in\overline{B_0}}  \psi^*_{k\mu}\psi^*_{k\nu}(\psi^2_{k1}+\psi^2_{k2N}),
\end{split}
\ee
where $\|\cdot\|$ is the spectral norm.
\end{lem}

This conclusion leads us to Eq.~\eqref{eq:map_phys_majorana} when the bulk modes are delocalized. Indeed, in this case $\kappa_{\mu\nu}=O(N^{-1})$, therefore the absolute value in Eq.~\eqref{correction} scales as $O(N^{-1/2})$ and the correction vanishes in the thermodynamic limit $N\to\infty$.

As an alternative to the theoretical prediction in Lemma 1, we explore the possibility of finding the correct angle $\alpha_0$ by optimization. To illustrate this method, we run the circuit for multiple angles and find the optimal value of $\alpha_0$. To estimate the braiding quality, we propose a cost function that favors a braided wavefunction if it is located at the opposite boundary of the chain. In particular, for braiding the left eigenmode, our cost function is given by
\be\label{eq:cost_function}
\begin{split}
\mathcal L(\alpha_0) = \sum_{x=1}^N  \Bigl(|\tilde\psi^L_{2x-1}|^2&+|\tilde\psi^L_{2x}|^2\Bigl)(N-x)^2,
\end{split}
\ee
where $\tilde\psi^{L,R}_\mu$ are the braided wavefunctions in Eq.~\eqref{eq:braided_wavefunction} for the angle $\alpha_0$.
 The cost functions for different runs on the same device are shown in Fig.~\ref{fig:fig3s} with the same parameters as in Fig.~\eqref{fig:fig3}.
We use a simple polynomial approximation to find the optimal angle $\alpha_0$ corresponding to the minimum of the approximation of the curve.
In our experiment, the optimal value ($0.256934\pi$) differs slightly from the theoretical value ($0.263127\pi$). This difference is due to the presence of disorder and noise, which modify the original Hamiltonian dynamics. 

\begin{figure}[t!]
    \centering
    \includegraphics[width=0.5\textwidth]{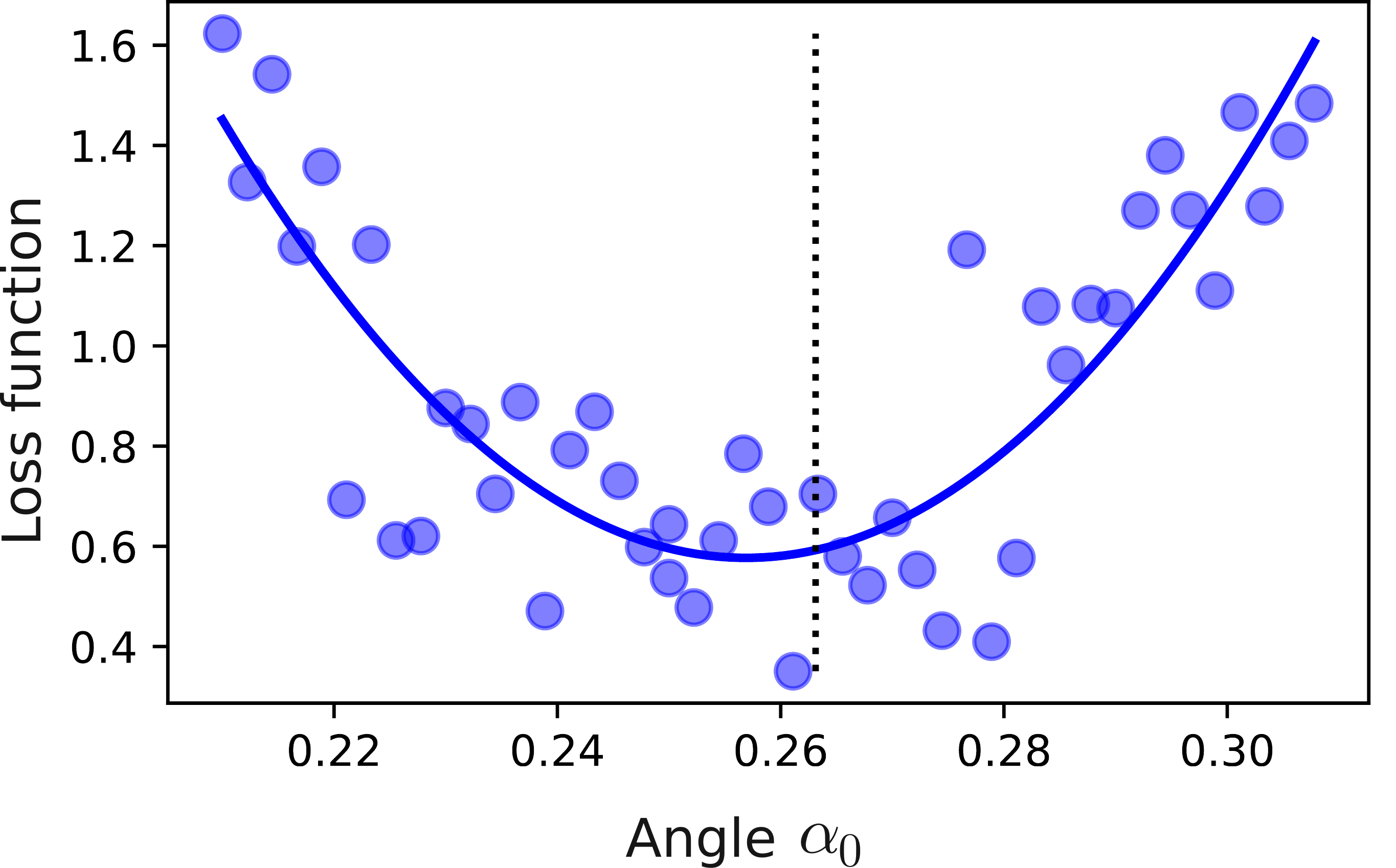}
    \caption{\textbf{Finding the angle by optimization}. Cost function in Eq.~\eqref{eq:cost_function} derived from experimental data (dots) approximated by a polynomial $I(\alpha_0)=a_1\alpha_0^2+a_2\alpha_0+a_3$ (solid curve), where $a_i$ are coefficients. The dotted vertical line shows the theoretically predicted value of the angle $\alpha_0$. We use a 5-qubit system on the $ibm\_hanoi$ device with parameters $\phi =\pi/16$, $\theta=\pi/4$, and $\varphi=0$ and maximum number of cycles $D=11$.}
    \label{fig:fig3s}
\end{figure}

Below we provide the proofs for both Lemma 1 and Lemma 2.\\

\textbf{Proof of Lemma 1}.
As a first step, we use the commutativity $[U_F,\Gamma_s] = 0$ for MZM operators and the decomposition in Eq.~\eqref{eqs:Gtog} to rewrite
\be\label{eqs:lem1_map}
\begin{split}
\mathcal E_{\alpha}(\Gamma_s) &= \limsum U^{\dag n}_FV^\dag(\alpha)\Gamma_sV(\alpha)U^n_F\\
& = \limsum \sum_{\mu=1}^{2N}\psi^s_\mu U^{\dag n}_F V^\dag(\alpha)\gamma_\mu V(\alpha)U^n_F.
\end{split}
\ee
where we define $V(\alpha) := \exp(-\alpha\gamma_1\gamma_{2N})$.
Now we can use the decomposition
$
V(\alpha) = \cos\alpha\, I - \sin\alpha\, \gamma_{1}\gamma_{2N}
$
and commutation relation between Majorana operators to express 
\be\label{eqs:gamma_map}
\begin{split}
V^\dag(\alpha)\gamma_\mu V(\alpha) = \gamma_\mu &- \delta_{\mu 1}(2\sin^2\alpha \gamma_1+\sin2\alpha \gamma_{2N})\\
&- \delta_{\mu 2N}(2\sin^2\alpha \gamma_{2N}-\sin2\alpha \gamma_{1}).
\end{split}
\ee
Combining this result with Eq.~\eqref{eqs:lem1_map}, we get
\be
\begin{split}
\mathcal E_{\alpha}(\Gamma_s) = \Gamma_s &-( 2\psi^s_1\sin^2\alpha-\psi^s_{2N}\sin2\alpha)\mathcal F_0(\gamma_1)\\
&-(2\psi^s_{2N}\sin^2\alpha+\psi^s_{1}\sin2\alpha)\mathcal F_0(\gamma_{2N}),
\end{split}
\ee
where we use the notation for the Fourier map from Eq.~\eqref{eqs:fourier_channel}.
Then, using the property in Eq.~\eqref{eqs:fourier}, we get the expression
\be
\begin{split}
\mathcal E_{\alpha}(\Gamma_s) = \Gamma_s&-\lim_{N\to\infty}\psi^L_1( 2\psi^s_1\sin^2\alpha-\psi^s_{2N}\sin2\alpha)\Gamma_L\\
&-\lim_{N\to\infty}\psi^R_{2N}(2\psi^s_{2N}\sin^2\alpha+\psi^s_{1}\sin2\alpha)\Gamma_R.
\end{split}
\ee
Inserting the particular values $s \in\{L,R\}$ and taking into account that $\psi^L_{2N}\sim \psi^R_1 \sim O(2^{-\Theta(N)})$ vanish as $N\to 0$, we get 
\be\label{eqs:lem1_almost_statement}
\begin{split}
&\mathcal E_{\alpha}(\Gamma_R) = (1-2(\psi^R_{2N})^2\sin^2\alpha)\Gamma_R+\sin2\alpha\, \psi^L_1\psi^R_{2N}\Gamma_L,\\
&\mathcal E_{\alpha}(\Gamma_L) = (1-2(\psi^L_{1})^2\sin^2\alpha)\Gamma_L-\sin2\alpha\, \psi^L_1\psi^R_{2N}\Gamma_R.\end{split}
\ee
By putting $\psi^L_1 = \psi^R_{2N} = \xi$, according to the Lemma's conditions, and choosing the angle
\be
\alpha\to \alpha_0 = \arcsin\frac{1}{\sqrt{2}\xi},
\ee
expression in Eq.~\eqref{eqs:lem1_almost_statement} converts into the statement of the Lemma. This step concludes our proof.\\

\textbf{Proof of Lemma 2}. Assuming that Majorana modes are the only single-fermion modes with zero frequency, we can write
\be
\gamma_\mu = \psi^L_\mu\Gamma_L+\psi^R_\mu\Gamma_R+\sum_{k\in \overline{B_0}}\psi^*_{k\mu}\Delta^0_k,
\ee
where $\overline{B_0}$ is the set of single-fermion modes in Eq.~\eqref{eq:single_fermion_modes} whose frequencies are distinct from zero. Then, the action of the target map on Majorana operator is
\be\label{eqs:lem2_map_expr}
\begin{split}
\mathcal E_{\alpha_0}(\gamma_\mu)& = \psi^L_\mu \mathcal E_{\alpha_0}(\Gamma_L)+\psi^R_\mu \mathcal E_{\alpha_0}(\Gamma_R)+\sum_{k\in \overline{B_0}}\psi^*_{k\mu}\mathcal E_{\alpha_0}(\Delta^0_k)\\
&= p(\psi^R_\mu \Gamma_L-\psi^L_\mu  \Gamma_R)+\sum_{k\in \overline{B_0}}\psi^*_{k\mu}\mathcal E_{\alpha_0}(\Delta^0_k),
\end{split}
\ee
where we used Lemma 1 to express the action of the target map on topological Majorana operators $\Gamma_{L,R}$.
The last term, in turn, can be evaluated using the explicit form of the map and Eq.~\eqref{eq:single_fermion_modes},
\be
\begin{split}
\mathcal E_\alpha&(\Delta_k^0) = \limsum e^{-i\omega^0_kn} U_F^{\dag n} V^\dag(\alpha)\Delta^0_kV(\alpha)U_F^n\\
&= \limsum \sum_{\nu=1}^{2N} \psi_{k\nu}e^{-i\omega^0_kn} U^{\dag n}_F V^\dag(\alpha)\gamma_\nu V(\alpha)U^n_F.
\end{split}
\ee
Using the result in Eq.~\eqref{eqs:gamma_map}, we get
\be\label{eqs:map_delta}
\begin{split}
\mathcal E_\alpha(\Delta_k^0) =& \limsum e^{-2i\omega^0_k n}\Delta^0_k\\
&-( 2\psi_{k1}\sin^2\alpha-\psi_{k2N}\sin2\alpha)\mathcal F_{-\omega^0_k}(\gamma_1)\\
&-(2\psi_{k2N}\sin^2\alpha+\psi_{k1}\sin2\alpha)\mathcal F_{-\omega^0_k}(\gamma_{2N}),
\end{split}
\ee
where we use the notation for the Fourier map from Eq.~\eqref{eq:fourier_components}. The first term of this expression vanishes for the set $k\in \overline{B_0}$. In turn, the action of the Fourier map is
\be
\begin{split}
\mathcal F_{-\omega_k}(\gamma_\mu) &= \sum_{k'\in\overline{B_0}}\psi^*_{k'\mu}\Delta^0_{k'}\limsum e^{-i(\omega^0_k+\omega^0_{k'})n}\\
& = \psi_{k\mu}\Delta_k^{0\dag},
\end{split}
\ee
where we used the fact that the wave-functions corresponding to opposite frequencies, i.e. $\omega^0_{k'} = -\omega^0_k$, satisfy $\psi_{k'\mu} = \psi_{k\mu}^*$ and $\Delta^0_{k'} = \Delta_k^{0\dag}$. Then, Eq.~\eqref{eqs:map_delta} takes the form
\be
\begin{split}
\mathcal E_{\alpha_0}(\Delta^0_k) 
&= -2\sin^2\alpha_0 (\psi_{k1}^2+\psi_{k2N}^2)\Delta_k^{0\dag} \\
&= -\frac 1{\xi^2}(\psi_{k1}^2+\psi_{k2N}^2)\Delta_k^{0\dag}.
\end{split}
\ee
Inserting this expression into Eq.~\eqref{eqs:lem2_map_expr}, we finally get
\be\label{eqs:E-rest}
\begin{split}
\mathcal E_{\alpha_0}(\gamma_\mu)
&=p(\psi^R_\mu \Gamma_L-\psi^L_\mu \Gamma_R)\\
&\quad-\frac{1}{\xi^2}\sum_{k\in\overline{B_0}}  \psi^*_{k\mu}(\psi_{k1}^2+\psi_{k2N}^2)\Delta_k^{0\dag}\\
& = p(\psi^R_\mu \Gamma_L-\psi^L_\mu \Gamma_R)\\
&\quad-\frac{1}{\xi^2}\sum_{k\in\overline{B_0}} \sum_{\nu=1}^{2N} \psi^*_{k\mu}\psi^*_{k\nu}(\psi_{k1}^2+\psi_{k2N}^2)\gamma_\nu\\
& = -\frac{1}{\xi^2}\sum_\nu \kappa_{\mu\nu}\gamma_\nu,
\end{split}
\ee
where $\kappa_{\mu\nu}$ are real numbers due to the fact that $\psi_{k\mu} = \psi^*_{k'\mu}$ for symmetric pairs $\omega_k^0 = -\omega_{k'}^0$ in $\overline{B_0}$. The maximum singular value of the operator in Eq.~\eqref{eqs:E-rest} is $\Lambda^{\rm max}_\mu =\xi^{-2}\sqrt{\sum_\nu\kappa_{\mu\nu}^2}$. This result leads us to the statement of the Lemma.


\begin{thebibliography}{53}%
\makeatletter
\providecommand \@ifxundefined [1]{%
 \@ifx{#1\undefined}
}%
\providecommand \@ifnum [1]{%
 \ifnum #1\expandafter \@firstoftwo
 \else \expandafter \@secondoftwo
 \fi
}%
\providecommand \@ifx [1]{%
 \ifx #1\expandafter \@firstoftwo
 \else \expandafter \@secondoftwo
 \fi
}%
\providecommand \natexlab [1]{#1}%
\providecommand \enquote  [1]{``#1''}%
\providecommand \bibnamefont  [1]{#1}%
\providecommand \bibfnamefont [1]{#1}%
\providecommand \citenamefont [1]{#1}%
\providecommand \href@noop [0]{\@secondoftwo}%
\providecommand \href [0]{\begingroup \@sanitize@url \@href}%
\providecommand \@href[1]{\@@startlink{#1}\@@href}%
\providecommand \@@href[1]{\endgroup#1\@@endlink}%
\providecommand \@sanitize@url [0]{\catcode `\\12\catcode `\$12\catcode
  `\&12\catcode `\#12\catcode `\^12\catcode `\_12\catcode `\%12\relax}%
\providecommand \@@startlink[1]{}%
\providecommand \@@endlink[0]{}%
\providecommand \url  [0]{\begingroup\@sanitize@url \@url }%
\providecommand \@url [1]{\endgroup\@href {#1}{\urlprefix }}%
\providecommand \urlprefix  [0]{URL }%
\providecommand \Eprint [0]{\href }%
\providecommand \doibase [0]{http://dx.doi.org/}%
\providecommand \selectlanguage [0]{\@gobble}%
\providecommand \bibinfo  [0]{\@secondoftwo}%
\providecommand \bibfield  [0]{\@secondoftwo}%
\providecommand \translation [1]{[#1]}%
\providecommand \BibitemOpen [0]{}%
\providecommand \bibitemStop [0]{}%
\providecommand \bibitemNoStop [0]{.\EOS\space}%
\providecommand \EOS [0]{\spacefactor3000\relax}%
\providecommand \BibitemShut  [1]{\csname bibitem#1\endcsname}%
\let\auto@bib@innerbib\@empty
\bibitem [{\citenamefont {Chow}\ \emph {et~al.}(2021)\citenamefont {Chow},
  \citenamefont {Dial},\ and\ \citenamefont {Gambetta}}]{chow2021ibm}%
  \BibitemOpen
  \bibfield  {author} {\bibinfo {author} {\bibfnamefont {J.}~\bibnamefont
  {Chow}}, \bibinfo {author} {\bibfnamefont {O.}~\bibnamefont {Dial}}, \ and\
  \bibinfo {author} {\bibfnamefont {J.}~\bibnamefont {Gambetta}},\ }\href
  {https://research.ibm.com/blog/127-qubit-quantum-processor-eagle?linkId=143497369&social_post=6027353319}
  {\enquote {\bibinfo {title} {Ibm quantum breaks the 100-qubit processor
  barrier},}\ } (\bibinfo {year} {2021})\BibitemShut {NoStop}%
\bibitem [{\citenamefont {Feynman}(1982)}]{feynman1982simulating}%
  \BibitemOpen
  \bibfield  {author} {\bibinfo {author} {\bibfnamefont {R.~P.}\ \bibnamefont
  {Feynman}},\ }\href {\doibase 10.1007/BF02650179} {\bibfield  {journal}
  {\bibinfo  {journal} {Int. J. Theor. Phys.}\ }\textbf {\bibinfo {volume}
  {21}},\ \bibinfo {pages} {467} (\bibinfo {year} {1982})}\BibitemShut
  {NoStop}%
\bibitem [{\citenamefont {Wen}(2017)}]{wen2017zoo}%
  \BibitemOpen
  \bibfield  {author} {\bibinfo {author} {\bibfnamefont {X.-G.}\ \bibnamefont
  {Wen}},\ }\href {\doibase 10.1103/RevModPhys.89.041004} {\bibfield  {journal}
  {\bibinfo  {journal} {Rev. Mod. Phys.}\ }\textbf {\bibinfo {volume} {89}},\
  \bibinfo {pages} {041004} (\bibinfo {year} {2017})}\BibitemShut {NoStop}%
\bibitem [{\citenamefont {Qi}\ and\ \citenamefont
  {Zhang}(2011)}]{qi2011topological}%
  \BibitemOpen
  \bibfield  {author} {\bibinfo {author} {\bibfnamefont {X.-L.}\ \bibnamefont
  {Qi}}\ and\ \bibinfo {author} {\bibfnamefont {S.-C.}\ \bibnamefont {Zhang}},\
  }\href {\doibase 10.1103/RevModPhys.83.1057} {\bibfield  {journal} {\bibinfo
  {journal} {Rev. Mod. Phys.}\ }\textbf {\bibinfo {volume} {83}},\ \bibinfo
  {pages} {1057} (\bibinfo {year} {2011})}\BibitemShut {NoStop}%
\bibitem [{\citenamefont {Aasen}\ \emph {et~al.}(2016)\citenamefont {Aasen},
  \citenamefont {Hell}, \citenamefont {Mishmash}, \citenamefont {Higginbotham},
  \citenamefont {Danon}, \citenamefont {Leijnse}, \citenamefont {Jespersen},
  \citenamefont {Folk}, \citenamefont {Marcus}, \citenamefont {Flensberg},\
  and\ \citenamefont {Alicea}}]{aasen2016milestones}%
  \BibitemOpen
  \bibfield  {author} {\bibinfo {author} {\bibfnamefont {D.}~\bibnamefont
  {Aasen}}, \bibinfo {author} {\bibfnamefont {M.}~\bibnamefont {Hell}},
  \bibinfo {author} {\bibfnamefont {R.~V.}\ \bibnamefont {Mishmash}}, \bibinfo
  {author} {\bibfnamefont {A.}~\bibnamefont {Higginbotham}}, \bibinfo {author}
  {\bibfnamefont {J.}~\bibnamefont {Danon}}, \bibinfo {author} {\bibfnamefont
  {M.}~\bibnamefont {Leijnse}}, \bibinfo {author} {\bibfnamefont {T.~S.}\
  \bibnamefont {Jespersen}}, \bibinfo {author} {\bibfnamefont {J.~A.}\
  \bibnamefont {Folk}}, \bibinfo {author} {\bibfnamefont {C.~M.}\ \bibnamefont
  {Marcus}}, \bibinfo {author} {\bibfnamefont {K.}~\bibnamefont {Flensberg}}, \
  and\ \bibinfo {author} {\bibfnamefont {J.}~\bibnamefont {Alicea}},\ }\href
  {\doibase 10.1103/PhysRevX.6.031016} {\bibfield  {journal} {\bibinfo
  {journal} {Phys. Rev. X}\ }\textbf {\bibinfo {volume} {6}},\ \bibinfo {pages}
  {031016} (\bibinfo {year} {2016})}\BibitemShut {NoStop}%
\bibitem [{\citenamefont {Lutchyn}\ \emph {et~al.}(2010)\citenamefont
  {Lutchyn}, \citenamefont {Sau},\ and\ \citenamefont
  {Das~Sarma}}]{lutchyn2010majorana}%
  \BibitemOpen
  \bibfield  {author} {\bibinfo {author} {\bibfnamefont {R.~M.}\ \bibnamefont
  {Lutchyn}}, \bibinfo {author} {\bibfnamefont {J.~D.}\ \bibnamefont {Sau}}, \
  and\ \bibinfo {author} {\bibfnamefont {S.}~\bibnamefont {Das~Sarma}},\ }\href
  {\doibase 10.1103/PhysRevLett.105.077001} {\bibfield  {journal} {\bibinfo
  {journal} {Phys. Rev. Lett.}\ }\textbf {\bibinfo {volume} {105}},\ \bibinfo
  {pages} {077001} (\bibinfo {year} {2010})}\BibitemShut {NoStop}%
\bibitem [{\citenamefont {Oreg}\ \emph {et~al.}(2010)\citenamefont {Oreg},
  \citenamefont {Refael},\ and\ \citenamefont {von Oppen}}]{oreg2010helical}%
  \BibitemOpen
  \bibfield  {author} {\bibinfo {author} {\bibfnamefont {Y.}~\bibnamefont
  {Oreg}}, \bibinfo {author} {\bibfnamefont {G.}~\bibnamefont {Refael}}, \ and\
  \bibinfo {author} {\bibfnamefont {F.}~\bibnamefont {von Oppen}},\ }\href
  {\doibase 10.1103/PhysRevLett.105.177002} {\bibfield  {journal} {\bibinfo
  {journal} {Phys. Rev. Lett.}\ }\textbf {\bibinfo {volume} {105}},\ \bibinfo
  {pages} {177002} (\bibinfo {year} {2010})}\BibitemShut {NoStop}%
\bibitem [{\citenamefont {Beenakker}(2013)}]{beenakker2013search}%
  \BibitemOpen
  \bibfield  {author} {\bibinfo {author} {\bibfnamefont {C.}~\bibnamefont
  {Beenakker}},\ }\href {\doibase 10.1146/annurev-conmatphys-030212-184337}
  {\bibfield  {journal} {\bibinfo  {journal} {Annu. Rev. Condens. Matter
  Phys.}\ }\textbf {\bibinfo {volume} {4}},\ \bibinfo {pages} {113} (\bibinfo
  {year} {2013})}\BibitemShut {NoStop}%
\bibitem [{\citenamefont {Lee}\ \emph {et~al.}(2014)\citenamefont {Lee},
  \citenamefont {Jiang}, \citenamefont {Houzet}, \citenamefont {Aguado},
  \citenamefont {Lieber},\ and\ \citenamefont {De~Franceschi}}]{lee2014spin}%
  \BibitemOpen
  \bibfield  {author} {\bibinfo {author} {\bibfnamefont {E.~J.~H.}\
  \bibnamefont {Lee}}, \bibinfo {author} {\bibfnamefont {X.}~\bibnamefont
  {Jiang}}, \bibinfo {author} {\bibfnamefont {M.}~\bibnamefont {Houzet}},
  \bibinfo {author} {\bibfnamefont {R.}~\bibnamefont {Aguado}}, \bibinfo
  {author} {\bibfnamefont {C.~M.}\ \bibnamefont {Lieber}}, \ and\ \bibinfo
  {author} {\bibfnamefont {S.}~\bibnamefont {De~Franceschi}},\ }\href {\doibase
  10.1038/nnano.2013.267} {\bibfield  {journal} {\bibinfo  {journal} {Nat.
  Nanotech.}\ }\textbf {\bibinfo {volume} {9}},\ \bibinfo {pages} {79}
  (\bibinfo {year} {2014})}\BibitemShut {NoStop}%
\bibitem [{\citenamefont {Kayyalha}\ \emph {et~al.}(2020)\citenamefont
  {Kayyalha}, \citenamefont {Xiao}, \citenamefont {Zhang}, \citenamefont
  {Shin}, \citenamefont {Jiang}, \citenamefont {Wang}, \citenamefont {Zhao},
  \citenamefont {Xiao}, \citenamefont {Zhang}, \citenamefont {Fijalkowski}
  \emph {et~al.}}]{kayyalha2020absence}%
  \BibitemOpen
  \bibfield  {author} {\bibinfo {author} {\bibfnamefont {M.}~\bibnamefont
  {Kayyalha}}, \bibinfo {author} {\bibfnamefont {D.}~\bibnamefont {Xiao}},
  \bibinfo {author} {\bibfnamefont {R.}~\bibnamefont {Zhang}}, \bibinfo
  {author} {\bibfnamefont {J.}~\bibnamefont {Shin}}, \bibinfo {author}
  {\bibfnamefont {J.}~\bibnamefont {Jiang}}, \bibinfo {author} {\bibfnamefont
  {F.}~\bibnamefont {Wang}}, \bibinfo {author} {\bibfnamefont {Y.-F.}\
  \bibnamefont {Zhao}}, \bibinfo {author} {\bibfnamefont {R.}~\bibnamefont
  {Xiao}}, \bibinfo {author} {\bibfnamefont {L.}~\bibnamefont {Zhang}},
  \bibinfo {author} {\bibfnamefont {K.~M.}\ \bibnamefont {Fijalkowski}},  \emph
  {et~al.},\ }\href {https://doi.org/10.1126/science.aax6361} {\bibfield
  {journal} {\bibinfo  {journal} {Science}\ }\textbf {\bibinfo {volume}
  {367}},\ \bibinfo {pages} {64} (\bibinfo {year} {2020})}\BibitemShut
  {NoStop}%
\bibitem [{\citenamefont {Valentini}\ \emph {et~al.}(2021)\citenamefont
  {Valentini}, \citenamefont {Pe{\~n}aranda}, \citenamefont {Hofmann},
  \citenamefont {Brauns}, \citenamefont {Hauschild}, \citenamefont {Krogstrup},
  \citenamefont {San-Jose}, \citenamefont {Prada}, \citenamefont {Aguado},\
  and\ \citenamefont {Katsaros}}]{valentini2021nontopological}%
  \BibitemOpen
  \bibfield  {author} {\bibinfo {author} {\bibfnamefont {M.}~\bibnamefont
  {Valentini}}, \bibinfo {author} {\bibfnamefont {F.}~\bibnamefont
  {Pe{\~n}aranda}}, \bibinfo {author} {\bibfnamefont {A.}~\bibnamefont
  {Hofmann}}, \bibinfo {author} {\bibfnamefont {M.}~\bibnamefont {Brauns}},
  \bibinfo {author} {\bibfnamefont {R.}~\bibnamefont {Hauschild}}, \bibinfo
  {author} {\bibfnamefont {P.}~\bibnamefont {Krogstrup}}, \bibinfo {author}
  {\bibfnamefont {P.}~\bibnamefont {San-Jose}}, \bibinfo {author}
  {\bibfnamefont {E.}~\bibnamefont {Prada}}, \bibinfo {author} {\bibfnamefont
  {R.}~\bibnamefont {Aguado}}, \ and\ \bibinfo {author} {\bibfnamefont
  {G.}~\bibnamefont {Katsaros}},\ }\href
  {https://doi.org/10.1126/science.abf1513} {\bibfield  {journal} {\bibinfo
  {journal} {Science}\ }\textbf {\bibinfo {volume} {373}},\ \bibinfo {pages}
  {82} (\bibinfo {year} {2021})}\BibitemShut {NoStop}%
\bibitem [{\citenamefont {Yu}\ \emph {et~al.}(2021)\citenamefont {Yu},
  \citenamefont {Chen}, \citenamefont {Gomanko}, \citenamefont {Badawy},
  \citenamefont {Bakkers}, \citenamefont {Zuo}, \citenamefont {Mourik},\ and\
  \citenamefont {Frolov}}]{yu2021non}%
  \BibitemOpen
  \bibfield  {author} {\bibinfo {author} {\bibfnamefont {P.}~\bibnamefont
  {Yu}}, \bibinfo {author} {\bibfnamefont {J.}~\bibnamefont {Chen}}, \bibinfo
  {author} {\bibfnamefont {M.}~\bibnamefont {Gomanko}}, \bibinfo {author}
  {\bibfnamefont {G.}~\bibnamefont {Badawy}}, \bibinfo {author} {\bibfnamefont
  {E.}~\bibnamefont {Bakkers}}, \bibinfo {author} {\bibfnamefont
  {K.}~\bibnamefont {Zuo}}, \bibinfo {author} {\bibfnamefont {V.}~\bibnamefont
  {Mourik}}, \ and\ \bibinfo {author} {\bibfnamefont {S.}~\bibnamefont
  {Frolov}},\ }\href {https://doi.org/10.1038/s41567-020-01107-w} {\bibfield
  {journal} {\bibinfo  {journal} {Nat. Phys.}\ }\textbf {\bibinfo {volume}
  {17}},\ \bibinfo {pages} {482} (\bibinfo {year} {2021})}\BibitemShut
  {NoStop}%
\bibitem [{\citenamefont {Salda{\~n}a}\ \emph {et~al.}(2021)\citenamefont
  {Salda{\~n}a}, \citenamefont {Vekris}, \citenamefont {Pave{\v{s}}i{\v{c}}},
  \citenamefont {Krogstrup}, \citenamefont {{\v{Z}}itko}, \citenamefont
  {Grove-Rasmussen},\ and\ \citenamefont {Nyg{\aa}rd}}]{saldana2021coulombic}%
  \BibitemOpen
  \bibfield  {author} {\bibinfo {author} {\bibfnamefont {J.~C.~E.}\
  \bibnamefont {Salda{\~n}a}}, \bibinfo {author} {\bibfnamefont
  {A.}~\bibnamefont {Vekris}}, \bibinfo {author} {\bibfnamefont
  {L.}~\bibnamefont {Pave{\v{s}}i{\v{c}}}}, \bibinfo {author} {\bibfnamefont
  {P.}~\bibnamefont {Krogstrup}}, \bibinfo {author} {\bibfnamefont
  {R.}~\bibnamefont {{\v{Z}}itko}}, \bibinfo {author} {\bibfnamefont
  {K.}~\bibnamefont {Grove-Rasmussen}}, \ and\ \bibinfo {author} {\bibfnamefont
  {J.}~\bibnamefont {Nyg{\aa}rd}},\ }\href {https://arxiv.org/abs/2101.10794v2}
  {\bibfield  {journal} {\bibinfo  {journal} {arXiv:2101.10794}\ } (\bibinfo
  {year} {2021})}\BibitemShut {NoStop}%
\bibitem [{\citenamefont {Wang}\ \emph {et~al.}(2021)\citenamefont {Wang},
  \citenamefont {Wiebe}, \citenamefont {Zhong}, \citenamefont {Gu},\ and\
  \citenamefont {Wiesendanger}}]{wang2021spin}%
  \BibitemOpen
  \bibfield  {author} {\bibinfo {author} {\bibfnamefont {D.}~\bibnamefont
  {Wang}}, \bibinfo {author} {\bibfnamefont {J.}~\bibnamefont {Wiebe}},
  \bibinfo {author} {\bibfnamefont {R.}~\bibnamefont {Zhong}}, \bibinfo
  {author} {\bibfnamefont {G.}~\bibnamefont {Gu}}, \ and\ \bibinfo {author}
  {\bibfnamefont {R.}~\bibnamefont {Wiesendanger}},\ }\href {\doibase
  10.1103/PhysRevLett.126.076802} {\bibfield  {journal} {\bibinfo  {journal}
  {Phys. Rev. Lett.}\ }\textbf {\bibinfo {volume} {126}},\ \bibinfo {pages}
  {076802} (\bibinfo {year} {2021})}\BibitemShut {NoStop}%
\bibitem [{\citenamefont {Khemani}\ \emph {et~al.}(2016)\citenamefont
  {Khemani}, \citenamefont {Lazarides}, \citenamefont {Moessner},\ and\
  \citenamefont {Sondhi}}]{khemani2016phase}%
  \BibitemOpen
  \bibfield  {author} {\bibinfo {author} {\bibfnamefont {V.}~\bibnamefont
  {Khemani}}, \bibinfo {author} {\bibfnamefont {A.}~\bibnamefont {Lazarides}},
  \bibinfo {author} {\bibfnamefont {R.}~\bibnamefont {Moessner}}, \ and\
  \bibinfo {author} {\bibfnamefont {S.~L.}\ \bibnamefont {Sondhi}},\ }\href
  {\doibase 10.1103/PhysRevLett.116.250401} {\bibfield  {journal} {\bibinfo
  {journal} {Phys. Rev. Lett.}\ }\textbf {\bibinfo {volume} {116}},\ \bibinfo
  {pages} {250401} (\bibinfo {year} {2016})}\BibitemShut {NoStop}%
\bibitem [{\citenamefont {Liu}\ \emph {et~al.}(2013)\citenamefont {Liu},
  \citenamefont {Levchenko},\ and\ \citenamefont {Baranger}}]{liu2013floquet}%
  \BibitemOpen
  \bibfield  {author} {\bibinfo {author} {\bibfnamefont {D.~E.}\ \bibnamefont
  {Liu}}, \bibinfo {author} {\bibfnamefont {A.}~\bibnamefont {Levchenko}}, \
  and\ \bibinfo {author} {\bibfnamefont {H.~U.}\ \bibnamefont {Baranger}},\
  }\href {\doibase 10.1103/PhysRevLett.111.047002} {\bibfield  {journal}
  {\bibinfo  {journal} {Phys. Rev. Lett.}\ }\textbf {\bibinfo {volume} {111}},\
  \bibinfo {pages} {047002} (\bibinfo {year} {2013})}\BibitemShut {NoStop}%
\bibitem [{\citenamefont {Potter}\ \emph {et~al.}(2016)\citenamefont {Potter},
  \citenamefont {Morimoto},\ and\ \citenamefont
  {Vishwanath}}]{potter2016classification}%
  \BibitemOpen
  \bibfield  {author} {\bibinfo {author} {\bibfnamefont {A.~C.}\ \bibnamefont
  {Potter}}, \bibinfo {author} {\bibfnamefont {T.}~\bibnamefont {Morimoto}}, \
  and\ \bibinfo {author} {\bibfnamefont {A.}~\bibnamefont {Vishwanath}},\
  }\href {\doibase 10.1103/PhysRevX.6.041001} {\bibfield  {journal} {\bibinfo
  {journal} {Phys. Rev. X}\ }\textbf {\bibinfo {volume} {6}},\ \bibinfo {pages}
  {041001} (\bibinfo {year} {2016})}\BibitemShut {NoStop}%
\bibitem [{\citenamefont {Levitov}\ \emph {et~al.}(2001)\citenamefont
  {Levitov}, \citenamefont {Orlando}, \citenamefont {Majer},\ and\
  \citenamefont {Mooij}}]{levitov2001quantum}%
  \BibitemOpen
  \bibfield  {author} {\bibinfo {author} {\bibfnamefont {L.}~\bibnamefont
  {Levitov}}, \bibinfo {author} {\bibfnamefont {T.}~\bibnamefont {Orlando}},
  \bibinfo {author} {\bibfnamefont {J.}~\bibnamefont {Majer}}, \ and\ \bibinfo
  {author} {\bibfnamefont {J.}~\bibnamefont {Mooij}},\ }\href
  {https://arxiv.org/abs/cond-mat/0108266} {\bibfield  {journal} {\bibinfo
  {journal} {arXiv:cond-mat/0108266}\ } (\bibinfo {year} {2001})}\BibitemShut
  {NoStop}%
\bibitem [{\citenamefont {You}\ \emph {et~al.}(2014)\citenamefont {You},
  \citenamefont {Wang}, \citenamefont {Zhang},\ and\ \citenamefont
  {Nori}}]{you2014encoding}%
  \BibitemOpen
  \bibfield  {author} {\bibinfo {author} {\bibfnamefont {J.}~\bibnamefont
  {You}}, \bibinfo {author} {\bibfnamefont {Z.}~\bibnamefont {Wang}}, \bibinfo
  {author} {\bibfnamefont {W.}~\bibnamefont {Zhang}}, \ and\ \bibinfo {author}
  {\bibfnamefont {F.}~\bibnamefont {Nori}},\ }\href
  {https://doi.org/10.1038/srep05535} {\bibfield  {journal} {\bibinfo
  {journal} {Sci. Rep.}\ }\textbf {\bibinfo {volume} {4}},\ \bibinfo {pages}
  {1} (\bibinfo {year} {2014})}\BibitemShut {NoStop}%
\bibitem [{\citenamefont {Backens}\ \emph {et~al.}(2017)\citenamefont
  {Backens}, \citenamefont {Shnirman}, \citenamefont {Makhlin}, \citenamefont
  {Gefen}, \citenamefont {Mooij},\ and\ \citenamefont
  {Sch\"on}}]{backens2017emulating}%
  \BibitemOpen
  \bibfield  {author} {\bibinfo {author} {\bibfnamefont {S.}~\bibnamefont
  {Backens}}, \bibinfo {author} {\bibfnamefont {A.}~\bibnamefont {Shnirman}},
  \bibinfo {author} {\bibfnamefont {Y.}~\bibnamefont {Makhlin}}, \bibinfo
  {author} {\bibfnamefont {Y.}~\bibnamefont {Gefen}}, \bibinfo {author}
  {\bibfnamefont {J.~E.}\ \bibnamefont {Mooij}}, \ and\ \bibinfo {author}
  {\bibfnamefont {G.}~\bibnamefont {Sch\"on}},\ }\href {\doibase
  10.1103/PhysRevB.96.195402} {\bibfield  {journal} {\bibinfo  {journal} {Phys.
  Rev. B}\ }\textbf {\bibinfo {volume} {96}},\ \bibinfo {pages} {195402}
  (\bibinfo {year} {2017})}\BibitemShut {NoStop}%
\bibitem [{\citenamefont {Kitagawa}\ \emph {et~al.}(2012)\citenamefont
  {Kitagawa}, \citenamefont {Broome}, \citenamefont {Fedrizzi}, \citenamefont
  {Rudner}, \citenamefont {Berg}, \citenamefont {Kassal}, \citenamefont
  {Aspuru-Guzik}, \citenamefont {Demler},\ and\ \citenamefont
  {White}}]{kitagawa2012observation}%
  \BibitemOpen
  \bibfield  {author} {\bibinfo {author} {\bibfnamefont {T.}~\bibnamefont
  {Kitagawa}}, \bibinfo {author} {\bibfnamefont {M.~A.}\ \bibnamefont
  {Broome}}, \bibinfo {author} {\bibfnamefont {A.}~\bibnamefont {Fedrizzi}},
  \bibinfo {author} {\bibfnamefont {M.~S.}\ \bibnamefont {Rudner}}, \bibinfo
  {author} {\bibfnamefont {E.}~\bibnamefont {Berg}}, \bibinfo {author}
  {\bibfnamefont {I.}~\bibnamefont {Kassal}}, \bibinfo {author} {\bibfnamefont
  {A.}~\bibnamefont {Aspuru-Guzik}}, \bibinfo {author} {\bibfnamefont
  {E.}~\bibnamefont {Demler}}, \ and\ \bibinfo {author} {\bibfnamefont {A.~G.}\
  \bibnamefont {White}},\ }\href {https://doi.org/10.1038/ncomms1872}
  {\bibfield  {journal} {\bibinfo  {journal} {Nat. Commun.}\ }\textbf {\bibinfo
  {volume} {3}},\ \bibinfo {pages} {1} (\bibinfo {year} {2012})}\BibitemShut
  {NoStop}%
\bibitem [{\citenamefont {Cheng}\ \emph {et~al.}(2019)\citenamefont {Cheng},
  \citenamefont {Pan}, \citenamefont {Wang}, \citenamefont {Zhang},
  \citenamefont {Yu}, \citenamefont {Gover}, \citenamefont {Zhang},
  \citenamefont {Li}, \citenamefont {Zhou},\ and\ \citenamefont
  {Zhu}}]{cheng2019observation}%
  \BibitemOpen
  \bibfield  {author} {\bibinfo {author} {\bibfnamefont {Q.}~\bibnamefont
  {Cheng}}, \bibinfo {author} {\bibfnamefont {Y.}~\bibnamefont {Pan}}, \bibinfo
  {author} {\bibfnamefont {H.}~\bibnamefont {Wang}}, \bibinfo {author}
  {\bibfnamefont {C.}~\bibnamefont {Zhang}}, \bibinfo {author} {\bibfnamefont
  {D.}~\bibnamefont {Yu}}, \bibinfo {author} {\bibfnamefont {A.}~\bibnamefont
  {Gover}}, \bibinfo {author} {\bibfnamefont {H.}~\bibnamefont {Zhang}},
  \bibinfo {author} {\bibfnamefont {T.}~\bibnamefont {Li}}, \bibinfo {author}
  {\bibfnamefont {L.}~\bibnamefont {Zhou}}, \ and\ \bibinfo {author}
  {\bibfnamefont {S.}~\bibnamefont {Zhu}},\ }\href {\doibase
  10.1103/PhysRevLett.122.173901} {\bibfield  {journal} {\bibinfo  {journal}
  {Phys. Rev. Lett.}\ }\textbf {\bibinfo {volume} {122}},\ \bibinfo {pages}
  {173901} (\bibinfo {year} {2019})}\BibitemShut {NoStop}%
\bibitem [{\citenamefont {Xiao}\ \emph {et~al.}(2017)\citenamefont {Xiao},
  \citenamefont {Zhan}, \citenamefont {Bian}, \citenamefont {Wang},
  \citenamefont {Zhang}, \citenamefont {Wang}, \citenamefont {Li},
  \citenamefont {Mochizuki}, \citenamefont {Kim}, \citenamefont {Kawakami}
  \emph {et~al.}}]{xiao2017observation}%
  \BibitemOpen
  \bibfield  {author} {\bibinfo {author} {\bibfnamefont {L.}~\bibnamefont
  {Xiao}}, \bibinfo {author} {\bibfnamefont {X.}~\bibnamefont {Zhan}}, \bibinfo
  {author} {\bibfnamefont {Z.}~\bibnamefont {Bian}}, \bibinfo {author}
  {\bibfnamefont {K.}~\bibnamefont {Wang}}, \bibinfo {author} {\bibfnamefont
  {X.}~\bibnamefont {Zhang}}, \bibinfo {author} {\bibfnamefont
  {X.}~\bibnamefont {Wang}}, \bibinfo {author} {\bibfnamefont {J.}~\bibnamefont
  {Li}}, \bibinfo {author} {\bibfnamefont {K.}~\bibnamefont {Mochizuki}},
  \bibinfo {author} {\bibfnamefont {D.}~\bibnamefont {Kim}}, \bibinfo {author}
  {\bibfnamefont {N.}~\bibnamefont {Kawakami}},  \emph {et~al.},\ }\href
  {https://doi.org/10.1038/nphys4204} {\bibfield  {journal} {\bibinfo
  {journal} {Nat. Phys.}\ }\textbf {\bibinfo {volume} {13}},\ \bibinfo {pages}
  {1117} (\bibinfo {year} {2017})}\BibitemShut {NoStop}%
\bibitem [{\citenamefont {Smith}\ \emph
  {et~al.}(2019{\natexlab{a}})\citenamefont {Smith}, \citenamefont {Kim},
  \citenamefont {Pollmann},\ and\ \citenamefont
  {Knolle}}]{smith2019simulating}%
  \BibitemOpen
  \bibfield  {author} {\bibinfo {author} {\bibfnamefont {A.}~\bibnamefont
  {Smith}}, \bibinfo {author} {\bibfnamefont {M.}~\bibnamefont {Kim}}, \bibinfo
  {author} {\bibfnamefont {F.}~\bibnamefont {Pollmann}}, \ and\ \bibinfo
  {author} {\bibfnamefont {J.}~\bibnamefont {Knolle}},\ }\href
  {https://doi.org/10.1038/s41534-019-0217-0} {\bibfield  {journal} {\bibinfo
  {journal} {npj Quantum Inf.}\ }\textbf {\bibinfo {volume} {5}},\ \bibinfo
  {pages} {1} (\bibinfo {year} {2019}{\natexlab{a}})}\BibitemShut {NoStop}%
\bibitem [{\citenamefont {Tan}\ \emph {et~al.}(2019)\citenamefont {Tan},
  \citenamefont {Zhao}, \citenamefont {Liu}, \citenamefont {Xue}, \citenamefont
  {Yu}, \citenamefont {Wang},\ and\ \citenamefont {Yu}}]{tan2019simulation}%
  \BibitemOpen
  \bibfield  {author} {\bibinfo {author} {\bibfnamefont {X.}~\bibnamefont
  {Tan}}, \bibinfo {author} {\bibfnamefont {Y.~X.}\ \bibnamefont {Zhao}},
  \bibinfo {author} {\bibfnamefont {Q.}~\bibnamefont {Liu}}, \bibinfo {author}
  {\bibfnamefont {G.}~\bibnamefont {Xue}}, \bibinfo {author} {\bibfnamefont
  {H.-F.}\ \bibnamefont {Yu}}, \bibinfo {author} {\bibfnamefont {Z.~D.}\
  \bibnamefont {Wang}}, \ and\ \bibinfo {author} {\bibfnamefont
  {Y.}~\bibnamefont {Yu}},\ }\href {\doibase 10.1103/PhysRevLett.122.010501}
  {\bibfield  {journal} {\bibinfo  {journal} {Phys. Rev. Lett.}\ }\textbf
  {\bibinfo {volume} {122}},\ \bibinfo {pages} {010501} (\bibinfo {year}
  {2019})}\BibitemShut {NoStop}%
\bibitem [{\citenamefont {Fauseweh}\ and\ \citenamefont
  {Zhu}(2021)}]{fauseweh2021digital}%
  \BibitemOpen
  \bibfield  {author} {\bibinfo {author} {\bibfnamefont {B.}~\bibnamefont
  {Fauseweh}}\ and\ \bibinfo {author} {\bibfnamefont {J.-X.}\ \bibnamefont
  {Zhu}},\ }\href {https://doi.org/10.1007/s11128-021-03079-z} {\bibfield
  {journal} {\bibinfo  {journal} {Quantum Inf. Process.}\ }\textbf {\bibinfo
  {volume} {20}},\ \bibinfo {pages} {1} (\bibinfo {year} {2021})}\BibitemShut
  {NoStop}%
\bibitem [{\citenamefont {Bassman}\ \emph {et~al.}(2021)\citenamefont
  {Bassman}, \citenamefont {Urbanek}, \citenamefont {Metcalf}, \citenamefont
  {Carter}, \citenamefont {Kemper},\ and\ \citenamefont
  {de~Jong}}]{bassman2021simulating}%
  \BibitemOpen
  \bibfield  {author} {\bibinfo {author} {\bibfnamefont {L.}~\bibnamefont
  {Bassman}}, \bibinfo {author} {\bibfnamefont {M.}~\bibnamefont {Urbanek}},
  \bibinfo {author} {\bibfnamefont {M.}~\bibnamefont {Metcalf}}, \bibinfo
  {author} {\bibfnamefont {J.}~\bibnamefont {Carter}}, \bibinfo {author}
  {\bibfnamefont {A.~F.}\ \bibnamefont {Kemper}}, \ and\ \bibinfo {author}
  {\bibfnamefont {W.~A.}\ \bibnamefont {de~Jong}},\ }\href {\doibase
  10.1088/2058-9565/ac1ca6} {\bibfield  {journal} {\bibinfo  {journal} {Quantum
  Sci. Technol.}\ }\textbf {\bibinfo {volume} {6}},\ \bibinfo {pages} {043002}
  (\bibinfo {year} {2021})}\BibitemShut {NoStop}%
\bibitem [{\citenamefont {Koh}\ \emph {et~al.}(2022)\citenamefont {Koh},
  \citenamefont {Tai}, \citenamefont {Phee}, \citenamefont {Ng},\ and\
  \citenamefont {Lee}}]{koh2021stabilizing}%
  \BibitemOpen
  \bibfield  {author} {\bibinfo {author} {\bibfnamefont {J.~M.}\ \bibnamefont
  {Koh}}, \bibinfo {author} {\bibfnamefont {T.}~\bibnamefont {Tai}}, \bibinfo
  {author} {\bibfnamefont {Y.~H.}\ \bibnamefont {Phee}}, \bibinfo {author}
  {\bibfnamefont {W.~E.}\ \bibnamefont {Ng}}, \ and\ \bibinfo {author}
  {\bibfnamefont {C.~H.}\ \bibnamefont {Lee}},\ }\href {\doibase
  10.1038/s41534-022-00527-1} {\bibfield  {journal} {\bibinfo  {journal} {npj
  Quantum. Inf.}\ }\textbf {\bibinfo {volume} {8}},\ \bibinfo {pages} {16}
  (\bibinfo {year} {2022})}\BibitemShut {NoStop}%
\bibitem [{\citenamefont {Smith}\ \emph
  {et~al.}(2019{\natexlab{b}})\citenamefont {Smith}, \citenamefont {Jobst},
  \citenamefont {Green},\ and\ \citenamefont {Pollmann}}]{smith2019crossing}%
  \BibitemOpen
  \bibfield  {author} {\bibinfo {author} {\bibfnamefont {A.}~\bibnamefont
  {Smith}}, \bibinfo {author} {\bibfnamefont {B.}~\bibnamefont {Jobst}},
  \bibinfo {author} {\bibfnamefont {A.~G.}\ \bibnamefont {Green}}, \ and\
  \bibinfo {author} {\bibfnamefont {F.}~\bibnamefont {Pollmann}},\ }\href
  {https://arxiv.org/abs/1910.05351} {\bibfield  {journal} {\bibinfo  {journal}
  {arXiv:1910.05351}\ } (\bibinfo {year} {2019}{\natexlab{b}})}\BibitemShut
  {NoStop}%
\bibitem [{\citenamefont {Neill}\ \emph {et~al.}(2021)\citenamefont {Neill},
  \citenamefont {McCourt}, \citenamefont {Mi}, \citenamefont {Jiang},
  \citenamefont {Niu}, \citenamefont {Mruczkiewicz}, \citenamefont {Aleiner}
  \emph {et~al.}}]{neil2021accurately}%
  \BibitemOpen
  \bibfield  {author} {\bibinfo {author} {\bibfnamefont {C.}~\bibnamefont
  {Neill}}, \bibinfo {author} {\bibfnamefont {T.}~\bibnamefont {McCourt}},
  \bibinfo {author} {\bibfnamefont {X.}~\bibnamefont {Mi}}, \bibinfo {author}
  {\bibfnamefont {Z.}~\bibnamefont {Jiang}}, \bibinfo {author} {\bibfnamefont
  {M.~Y.}\ \bibnamefont {Niu}}, \bibinfo {author} {\bibfnamefont
  {W.}~\bibnamefont {Mruczkiewicz}}, \bibinfo {author} {\bibfnamefont
  {I.}~\bibnamefont {Aleiner}},  \emph {et~al.},\ }\href {\doibase
  10.1038/s41586-021-03576-2} {\bibfield  {journal} {\bibinfo  {journal}
  {Nature}\ }\textbf {\bibinfo {volume} {594}},\ \bibinfo {pages} {508}
  (\bibinfo {year} {2021})}\BibitemShut {NoStop}%
\bibitem [{\citenamefont {Wiekowski}\ \emph {et~al.}(2018)\citenamefont
  {Wiekowski}, \citenamefont {Ma{\'s}ka},\ and\ \citenamefont
  {Mierzejewski}}]{wikeckowski2018identification}%
  \BibitemOpen
  \bibfield  {author} {\bibinfo {author} {\bibfnamefont {A.}~\bibnamefont
  {Wiekowski}}, \bibinfo {author} {\bibfnamefont {M.~M.}\ \bibnamefont
  {Ma{\'s}ka}}, \ and\ \bibinfo {author} {\bibfnamefont {M.}~\bibnamefont
  {Mierzejewski}},\ }\href {\doibase 10.1103/PhysRevLett.120.040504} {\bibfield
   {journal} {\bibinfo  {journal} {Phys. Rev. Lett.}\ }\textbf {\bibinfo
  {volume} {120}},\ \bibinfo {pages} {040504} (\bibinfo {year}
  {2018})}\BibitemShut {NoStop}%
\bibitem [{\citenamefont {Shtanko}\ and\ \citenamefont
  {Movassagh}(2020)}]{shtanko2020unitary}%
  \BibitemOpen
  \bibfield  {author} {\bibinfo {author} {\bibfnamefont {O.}~\bibnamefont
  {Shtanko}}\ and\ \bibinfo {author} {\bibfnamefont {R.}~\bibnamefont
  {Movassagh}},\ }\href {\doibase 10.1103/PhysRevLett.125.086804} {\bibfield
  {journal} {\bibinfo  {journal} {Phys. Rev. Lett.}\ }\textbf {\bibinfo
  {volume} {125}},\ \bibinfo {pages} {086804} (\bibinfo {year}
  {2020})}\BibitemShut {NoStop}%
\bibitem [{\citenamefont {Azses}\ \emph {et~al.}(2020)\citenamefont {Azses},
  \citenamefont {Haenel}, \citenamefont {Naveh}, \citenamefont {Raussendorf},
  \citenamefont {Sela},\ and\ \citenamefont
  {Dalla~Torre}}]{azses2020identification}%
  \BibitemOpen
  \bibfield  {author} {\bibinfo {author} {\bibfnamefont {D.}~\bibnamefont
  {Azses}}, \bibinfo {author} {\bibfnamefont {R.}~\bibnamefont {Haenel}},
  \bibinfo {author} {\bibfnamefont {Y.}~\bibnamefont {Naveh}}, \bibinfo
  {author} {\bibfnamefont {R.}~\bibnamefont {Raussendorf}}, \bibinfo {author}
  {\bibfnamefont {E.}~\bibnamefont {Sela}}, \ and\ \bibinfo {author}
  {\bibfnamefont {E.~G.}\ \bibnamefont {Dalla~Torre}},\ }\href {\doibase
  10.1103/PhysRevLett.125.120502} {\bibfield  {journal} {\bibinfo  {journal}
  {Phys. Rev. Lett.}\ }\textbf {\bibinfo {volume} {125}},\ \bibinfo {pages}
  {120502} (\bibinfo {year} {2020})}\BibitemShut {NoStop}%
\bibitem [{\citenamefont {Choo}\ \emph {et~al.}(2018)\citenamefont {Choo},
  \citenamefont {von Keyserlingk}, \citenamefont {Regnault},\ and\
  \citenamefont {Neupert}}]{choo2018measurement}%
  \BibitemOpen
  \bibfield  {author} {\bibinfo {author} {\bibfnamefont {K.}~\bibnamefont
  {Choo}}, \bibinfo {author} {\bibfnamefont {C.~W.}\ \bibnamefont {von
  Keyserlingk}}, \bibinfo {author} {\bibfnamefont {N.}~\bibnamefont
  {Regnault}}, \ and\ \bibinfo {author} {\bibfnamefont {T.}~\bibnamefont
  {Neupert}},\ }\href {\doibase 10.1103/PhysRevLett.121.086808} {\bibfield
  {journal} {\bibinfo  {journal} {Phys. Rev. Lett.}\ }\textbf {\bibinfo
  {volume} {121}},\ \bibinfo {pages} {086808} (\bibinfo {year}
  {2018})}\BibitemShut {NoStop}%
\bibitem [{\citenamefont {Zhang}\ \emph {et~al.}(2022)\citenamefont {Zhang},
  \citenamefont {Jiang}, \citenamefont {Deng}, \citenamefont {Wang},
  \citenamefont {Chen}, \citenamefont {Zhang}, \citenamefont {Ren},
  \citenamefont {Dong}, \citenamefont {Xu}, \citenamefont {Gao} \emph
  {et~al.}}]{zhang2022digital}%
  \BibitemOpen
  \bibfield  {author} {\bibinfo {author} {\bibfnamefont {X.}~\bibnamefont
  {Zhang}}, \bibinfo {author} {\bibfnamefont {W.}~\bibnamefont {Jiang}},
  \bibinfo {author} {\bibfnamefont {J.}~\bibnamefont {Deng}}, \bibinfo {author}
  {\bibfnamefont {K.}~\bibnamefont {Wang}}, \bibinfo {author} {\bibfnamefont
  {J.}~\bibnamefont {Chen}}, \bibinfo {author} {\bibfnamefont {P.}~\bibnamefont
  {Zhang}}, \bibinfo {author} {\bibfnamefont {W.}~\bibnamefont {Ren}}, \bibinfo
  {author} {\bibfnamefont {H.}~\bibnamefont {Dong}}, \bibinfo {author}
  {\bibfnamefont {S.}~\bibnamefont {Xu}}, \bibinfo {author} {\bibfnamefont
  {Y.}~\bibnamefont {Gao}},  \emph {et~al.},\ }\href
  {https://doi.org/10.1038/s41586-022-04854-3} {\bibfield  {journal} {\bibinfo
  {journal} {Nature}\ }\textbf {\bibinfo {volume} {607}},\ \bibinfo {pages}
  {468} (\bibinfo {year} {2022})}\BibitemShut {NoStop}%
\bibitem [{\citenamefont {Xu}\ \emph {et~al.}(2016)\citenamefont {Xu},
  \citenamefont {Sun}, \citenamefont {Han}, \citenamefont {Li}, \citenamefont
  {Pachos},\ and\ \citenamefont {Guo}}]{xu2016simulating}%
  \BibitemOpen
  \bibfield  {author} {\bibinfo {author} {\bibfnamefont {J.-S.}\ \bibnamefont
  {Xu}}, \bibinfo {author} {\bibfnamefont {K.}~\bibnamefont {Sun}}, \bibinfo
  {author} {\bibfnamefont {Y.-J.}\ \bibnamefont {Han}}, \bibinfo {author}
  {\bibfnamefont {C.-F.}\ \bibnamefont {Li}}, \bibinfo {author} {\bibfnamefont
  {J.~K.}\ \bibnamefont {Pachos}}, \ and\ \bibinfo {author} {\bibfnamefont
  {G.-C.}\ \bibnamefont {Guo}},\ }\href {https://doi.org/10.1038/ncomms13194}
  {\bibfield  {journal} {\bibinfo  {journal} {Nat. Commun.}\ }\textbf {\bibinfo
  {volume} {7}},\ \bibinfo {pages} {1} (\bibinfo {year} {2016})}\BibitemShut
  {NoStop}%
\bibitem [{\citenamefont {Xu}\ \emph {et~al.}(2018)\citenamefont {Xu},
  \citenamefont {Sun}, \citenamefont {Pachos}, \citenamefont {Han},
  \citenamefont {Li},\ and\ \citenamefont {Guo}}]{xu2018photonic}%
  \BibitemOpen
  \bibfield  {author} {\bibinfo {author} {\bibfnamefont {J.-S.}\ \bibnamefont
  {Xu}}, \bibinfo {author} {\bibfnamefont {K.}~\bibnamefont {Sun}}, \bibinfo
  {author} {\bibfnamefont {J.~K.}\ \bibnamefont {Pachos}}, \bibinfo {author}
  {\bibfnamefont {Y.-J.}\ \bibnamefont {Han}}, \bibinfo {author} {\bibfnamefont
  {C.-F.}\ \bibnamefont {Li}}, \ and\ \bibinfo {author} {\bibfnamefont {G.-C.}\
  \bibnamefont {Guo}},\ }\href {https://doi.org/10.1126/sciadv.aat6533}
  {\bibfield  {journal} {\bibinfo  {journal} {Sci. Adv.}\ }\textbf {\bibinfo
  {volume} {4}},\ \bibinfo {pages} {eaat6533} (\bibinfo {year}
  {2018})}\BibitemShut {NoStop}%
\bibitem [{\citenamefont {Liu}\ \emph {et~al.}(2021)\citenamefont {Liu},
  \citenamefont {Sun}, \citenamefont {Pachos}, \citenamefont {Yang},
  \citenamefont {Meng}, \citenamefont {Liao}, \citenamefont {Li}, \citenamefont
  {Wang}, \citenamefont {Luo}, \citenamefont {He}, \citenamefont {Huang},
  \citenamefont {Ding}, \citenamefont {Xu}, \citenamefont {Han}, \citenamefont
  {Li},\ and\ \citenamefont {Guo}}]{liu2021topological}%
  \BibitemOpen
  \bibfield  {author} {\bibinfo {author} {\bibfnamefont {Z.-H.}\ \bibnamefont
  {Liu}}, \bibinfo {author} {\bibfnamefont {K.}~\bibnamefont {Sun}}, \bibinfo
  {author} {\bibfnamefont {J.~K.}\ \bibnamefont {Pachos}}, \bibinfo {author}
  {\bibfnamefont {M.}~\bibnamefont {Yang}}, \bibinfo {author} {\bibfnamefont
  {Y.}~\bibnamefont {Meng}}, \bibinfo {author} {\bibfnamefont {Y.-W.}\
  \bibnamefont {Liao}}, \bibinfo {author} {\bibfnamefont {Q.}~\bibnamefont
  {Li}}, \bibinfo {author} {\bibfnamefont {J.-F.}\ \bibnamefont {Wang}},
  \bibinfo {author} {\bibfnamefont {Z.-Y.}\ \bibnamefont {Luo}}, \bibinfo
  {author} {\bibfnamefont {Y.-F.}\ \bibnamefont {He}}, \bibinfo {author}
  {\bibfnamefont {D.-Y.}\ \bibnamefont {Huang}}, \bibinfo {author}
  {\bibfnamefont {G.-R.}\ \bibnamefont {Ding}}, \bibinfo {author}
  {\bibfnamefont {J.-S.}\ \bibnamefont {Xu}}, \bibinfo {author} {\bibfnamefont
  {Y.-J.}\ \bibnamefont {Han}}, \bibinfo {author} {\bibfnamefont {C.-F.}\
  \bibnamefont {Li}}, \ and\ \bibinfo {author} {\bibfnamefont {G.-C.}\
  \bibnamefont {Guo}},\ }\href {\doibase 10.1103/PRXQuantum.2.030323}
  {\bibfield  {journal} {\bibinfo  {journal} {PRX Quantum}\ }\textbf {\bibinfo
  {volume} {2}},\ \bibinfo {pages} {030323} (\bibinfo {year}
  {2021})}\BibitemShut {NoStop}%
\bibitem [{\citenamefont {Wootton}(2017)}]{wootton2017demonstrating}%
  \BibitemOpen
  \bibfield  {author} {\bibinfo {author} {\bibfnamefont {J.~R.}\ \bibnamefont
  {Wootton}},\ }\href {https://doi.org/10.1088/2058-9565/aa5c73} {\bibfield
  {journal} {\bibinfo  {journal} {Quantum Sci. Technol.}\ }\textbf {\bibinfo
  {volume} {2}},\ \bibinfo {pages} {015006} (\bibinfo {year}
  {2017})}\BibitemShut {NoStop}%
\bibitem [{\citenamefont {Zhong}\ \emph {et~al.}(2016)\citenamefont {Zhong},
  \citenamefont {Xu}, \citenamefont {Wang}, \citenamefont {Song}, \citenamefont
  {Guo}, \citenamefont {Liu}, \citenamefont {Xu}, \citenamefont {Xia},
  \citenamefont {Lu}, \citenamefont {Han}, \citenamefont {Pan},\ and\
  \citenamefont {Wang}}]{zhong2016emulating}%
  \BibitemOpen
  \bibfield  {author} {\bibinfo {author} {\bibfnamefont {Y.~P.}\ \bibnamefont
  {Zhong}}, \bibinfo {author} {\bibfnamefont {D.}~\bibnamefont {Xu}}, \bibinfo
  {author} {\bibfnamefont {P.}~\bibnamefont {Wang}}, \bibinfo {author}
  {\bibfnamefont {C.}~\bibnamefont {Song}}, \bibinfo {author} {\bibfnamefont
  {Q.~J.}\ \bibnamefont {Guo}}, \bibinfo {author} {\bibfnamefont {W.~X.}\
  \bibnamefont {Liu}}, \bibinfo {author} {\bibfnamefont {K.}~\bibnamefont
  {Xu}}, \bibinfo {author} {\bibfnamefont {B.~X.}\ \bibnamefont {Xia}},
  \bibinfo {author} {\bibfnamefont {C.-Y.}\ \bibnamefont {Lu}}, \bibinfo
  {author} {\bibfnamefont {S.}~\bibnamefont {Han}}, \bibinfo {author}
  {\bibfnamefont {J.-W.}\ \bibnamefont {Pan}}, \ and\ \bibinfo {author}
  {\bibfnamefont {H.}~\bibnamefont {Wang}},\ }\href {\doibase
  10.1103/PhysRevLett.117.110501} {\bibfield  {journal} {\bibinfo  {journal}
  {Phys. Rev. Lett.}\ }\textbf {\bibinfo {volume} {117}},\ \bibinfo {pages}
  {110501} (\bibinfo {year} {2016})}\BibitemShut {NoStop}%
\bibitem [{\citenamefont {Song}\ \emph {et~al.}(2018)\citenamefont {Song},
  \citenamefont {Xu}, \citenamefont {Zhang}, \citenamefont {Wang},
  \citenamefont {Guo}, \citenamefont {Liu}, \citenamefont {Xu}, \citenamefont
  {Deng}, \citenamefont {Huang}, \citenamefont {Zheng}, \citenamefont {Zheng},
  \citenamefont {Wang}, \citenamefont {Zhu}, \citenamefont {Lu},\ and\
  \citenamefont {Pan}}]{song2018demonstration}%
  \BibitemOpen
  \bibfield  {author} {\bibinfo {author} {\bibfnamefont {C.}~\bibnamefont
  {Song}}, \bibinfo {author} {\bibfnamefont {D.}~\bibnamefont {Xu}}, \bibinfo
  {author} {\bibfnamefont {P.}~\bibnamefont {Zhang}}, \bibinfo {author}
  {\bibfnamefont {J.}~\bibnamefont {Wang}}, \bibinfo {author} {\bibfnamefont
  {Q.}~\bibnamefont {Guo}}, \bibinfo {author} {\bibfnamefont {W.}~\bibnamefont
  {Liu}}, \bibinfo {author} {\bibfnamefont {K.}~\bibnamefont {Xu}}, \bibinfo
  {author} {\bibfnamefont {H.}~\bibnamefont {Deng}}, \bibinfo {author}
  {\bibfnamefont {K.}~\bibnamefont {Huang}}, \bibinfo {author} {\bibfnamefont
  {D.}~\bibnamefont {Zheng}}, \bibinfo {author} {\bibfnamefont {S.-B.}\
  \bibnamefont {Zheng}}, \bibinfo {author} {\bibfnamefont {H.}~\bibnamefont
  {Wang}}, \bibinfo {author} {\bibfnamefont {X.}~\bibnamefont {Zhu}}, \bibinfo
  {author} {\bibfnamefont {C.-Y.}\ \bibnamefont {Lu}}, \ and\ \bibinfo {author}
  {\bibfnamefont {J.-W.}\ \bibnamefont {Pan}},\ }\href {\doibase
  10.1103/PhysRevLett.121.030502} {\bibfield  {journal} {\bibinfo  {journal}
  {Phys. Rev. Lett.}\ }\textbf {\bibinfo {volume} {121}},\ \bibinfo {pages}
  {030502} (\bibinfo {year} {2018})}\BibitemShut {NoStop}%
\bibitem [{\citenamefont {Huang}\ \emph {et~al.}(2021)\citenamefont {Huang},
  \citenamefont {Naro\ifmmode~\dot{z}\else \.{z}\fi{}niak}, \citenamefont
  {Liang}, \citenamefont {Zhao}, \citenamefont {Castellano}, \citenamefont
  {Gong}, \citenamefont {Wu}, \citenamefont {Wang}, \citenamefont {Lin},
  \citenamefont {Xu}, \citenamefont {Deng}, \citenamefont {Rong}, \citenamefont
  {Dowling}, \citenamefont {Peng}, \citenamefont {Byrnes}, \citenamefont
  {Zhu},\ and\ \citenamefont {Pan}}]{huang2021emulating}%
  \BibitemOpen
  \bibfield  {author} {\bibinfo {author} {\bibfnamefont {H.-L.}\ \bibnamefont
  {Huang}}, \bibinfo {author} {\bibfnamefont {M.}~\bibnamefont
  {Naro\ifmmode~\dot{z}\else \.{z}\fi{}niak}}, \bibinfo {author} {\bibfnamefont
  {F.}~\bibnamefont {Liang}}, \bibinfo {author} {\bibfnamefont
  {Y.}~\bibnamefont {Zhao}}, \bibinfo {author} {\bibfnamefont {A.~D.}\
  \bibnamefont {Castellano}}, \bibinfo {author} {\bibfnamefont
  {M.}~\bibnamefont {Gong}}, \bibinfo {author} {\bibfnamefont {Y.}~\bibnamefont
  {Wu}}, \bibinfo {author} {\bibfnamefont {S.}~\bibnamefont {Wang}}, \bibinfo
  {author} {\bibfnamefont {J.}~\bibnamefont {Lin}}, \bibinfo {author}
  {\bibfnamefont {Y.}~\bibnamefont {Xu}}, \bibinfo {author} {\bibfnamefont
  {H.}~\bibnamefont {Deng}}, \bibinfo {author} {\bibfnamefont {H.}~\bibnamefont
  {Rong}}, \bibinfo {author} {\bibfnamefont {J.~P.}\ \bibnamefont {Dowling}},
  \bibinfo {author} {\bibfnamefont {C.-Z.}\ \bibnamefont {Peng}}, \bibinfo
  {author} {\bibfnamefont {T.}~\bibnamefont {Byrnes}}, \bibinfo {author}
  {\bibfnamefont {X.}~\bibnamefont {Zhu}}, \ and\ \bibinfo {author}
  {\bibfnamefont {J.-W.}\ \bibnamefont {Pan}},\ }\href {\doibase
  10.1103/PhysRevLett.126.090502} {\bibfield  {journal} {\bibinfo  {journal}
  {Phys. Rev. Lett.}\ }\textbf {\bibinfo {volume} {126}},\ \bibinfo {pages}
  {090502} (\bibinfo {year} {2021})}\BibitemShut {NoStop}%
\bibitem [{\citenamefont {Rudner}\ and\ \citenamefont
  {Lindner}(2020)}]{rudner2020band}%
  \BibitemOpen
  \bibfield  {author} {\bibinfo {author} {\bibfnamefont {M.~S.}\ \bibnamefont
  {Rudner}}\ and\ \bibinfo {author} {\bibfnamefont {N.~H.}\ \bibnamefont
  {Lindner}},\ }\href {https://doi.org/10.1038/s42254-020-0170-z} {\bibfield
  {journal} {\bibinfo  {journal} {Nat. Rev. Phys.}\ }\textbf {\bibinfo {volume}
  {2}},\ \bibinfo {pages} {229} (\bibinfo {year} {2020})}\BibitemShut {NoStop}%
\bibitem [{\citenamefont {Shtanko}\ and\ \citenamefont
  {Movassagh}(2018)}]{shtanko2018stability}%
  \BibitemOpen
  \bibfield  {author} {\bibinfo {author} {\bibfnamefont {O.}~\bibnamefont
  {Shtanko}}\ and\ \bibinfo {author} {\bibfnamefont {R.}~\bibnamefont
  {Movassagh}},\ }\href {\doibase 10.1103/PhysRevLett.121.126803} {\bibfield
  {journal} {\bibinfo  {journal} {Phys. Rev. Lett.}\ }\textbf {\bibinfo
  {volume} {121}},\ \bibinfo {pages} {126803} (\bibinfo {year}
  {2018})}\BibitemShut {NoStop}%
\bibitem [{\citenamefont {Suzuki}\ \emph {et~al.}(2012)\citenamefont {Suzuki},
  \citenamefont {Inoue},\ and\ \citenamefont
  {Chakrabarti}}]{suzuki2012quantum}%
  \BibitemOpen
  \bibfield  {author} {\bibinfo {author} {\bibfnamefont {S.}~\bibnamefont
  {Suzuki}}, \bibinfo {author} {\bibfnamefont {J.}~\bibnamefont {Inoue}}, \
  and\ \bibinfo {author} {\bibfnamefont {B.}~\bibnamefont {Chakrabarti}},\
  }\href {https://books.google.com/books?id=y1S5BQAAQBAJ} {\emph {\bibinfo
  {title} {Quantum Ising Phases and Transitions in Transverse Ising Models}}},\
  Lecture Notes in Physics\ (\bibinfo  {publisher} {Springer Berlin
  Heidelberg},\ \bibinfo {year} {2012})\BibitemShut {NoStop}%
\bibitem [{\citenamefont {Jermyn}\ \emph {et~al.}(2014)\citenamefont {Jermyn},
  \citenamefont {Mong}, \citenamefont {Alicea},\ and\ \citenamefont
  {Fendley}}]{jermyn2014stability}%
  \BibitemOpen
  \bibfield  {author} {\bibinfo {author} {\bibfnamefont {A.~S.}\ \bibnamefont
  {Jermyn}}, \bibinfo {author} {\bibfnamefont {R.~S.~K.}\ \bibnamefont {Mong}},
  \bibinfo {author} {\bibfnamefont {J.}~\bibnamefont {Alicea}}, \ and\ \bibinfo
  {author} {\bibfnamefont {P.}~\bibnamefont {Fendley}},\ }\href {\doibase
  10.1103/PhysRevB.90.165106} {\bibfield  {journal} {\bibinfo  {journal} {Phys.
  Rev. B}\ }\textbf {\bibinfo {volume} {90}},\ \bibinfo {pages} {165106}
  (\bibinfo {year} {2014})}\BibitemShut {NoStop}%
\bibitem [{\citenamefont {Alicea}\ and\ \citenamefont
  {Fendley}(2016)}]{alicea2016topological}%
  \BibitemOpen
  \bibfield  {author} {\bibinfo {author} {\bibfnamefont {J.}~\bibnamefont
  {Alicea}}\ and\ \bibinfo {author} {\bibfnamefont {P.}~\bibnamefont
  {Fendley}},\ }\href {\doibase 10.1146/annurev-conmatphys-031115-011336}
  {\bibfield  {journal} {\bibinfo  {journal} {Annu. Rev. Condens. Matter
  Phys.}\ }\textbf {\bibinfo {volume} {7}},\ \bibinfo {pages} {119} (\bibinfo
  {year} {2016})}\BibitemShut {NoStop}%
\bibitem [{\citenamefont {Bomantara}\ and\ \citenamefont
  {Gong}(2018)}]{bomantara2018simulation}%
  \BibitemOpen
  \bibfield  {author} {\bibinfo {author} {\bibfnamefont {R.~W.}\ \bibnamefont
  {Bomantara}}\ and\ \bibinfo {author} {\bibfnamefont {J.}~\bibnamefont
  {Gong}},\ }\href {\doibase 10.1103/PhysRevLett.120.230405} {\bibfield
  {journal} {\bibinfo  {journal} {Phys. Rev. Lett.}\ }\textbf {\bibinfo
  {volume} {120}},\ \bibinfo {pages} {230405} (\bibinfo {year}
  {2018})}\BibitemShut {NoStop}%
\bibitem [{\citenamefont {Bauer}\ \emph {et~al.}(2019)\citenamefont {Bauer},
  \citenamefont {Pereg-Barnea}, \citenamefont {Karzig}, \citenamefont {Rieder},
  \citenamefont {Refael}, \citenamefont {Berg},\ and\ \citenamefont
  {Oreg}}]{bauer2019topologically}%
  \BibitemOpen
  \bibfield  {author} {\bibinfo {author} {\bibfnamefont {B.}~\bibnamefont
  {Bauer}}, \bibinfo {author} {\bibfnamefont {T.}~\bibnamefont {Pereg-Barnea}},
  \bibinfo {author} {\bibfnamefont {T.}~\bibnamefont {Karzig}}, \bibinfo
  {author} {\bibfnamefont {M.-T.}\ \bibnamefont {Rieder}}, \bibinfo {author}
  {\bibfnamefont {G.}~\bibnamefont {Refael}}, \bibinfo {author} {\bibfnamefont
  {E.}~\bibnamefont {Berg}}, \ and\ \bibinfo {author} {\bibfnamefont
  {Y.}~\bibnamefont {Oreg}},\ }\href {\doibase 10.1103/PhysRevB.100.041102}
  {\bibfield  {journal} {\bibinfo  {journal} {Phys. Rev. B}\ }\textbf {\bibinfo
  {volume} {100}},\ \bibinfo {pages} {041102} (\bibinfo {year}
  {2019})}\BibitemShut {NoStop}%
\bibitem [{\citenamefont {Abanin}\ \emph {et~al.}(2019)\citenamefont {Abanin},
  \citenamefont {Altman}, \citenamefont {Bloch},\ and\ \citenamefont
  {Serbyn}}]{abanin2019manybody}%
  \BibitemOpen
  \bibfield  {author} {\bibinfo {author} {\bibfnamefont {D.~A.}\ \bibnamefont
  {Abanin}}, \bibinfo {author} {\bibfnamefont {E.}~\bibnamefont {Altman}},
  \bibinfo {author} {\bibfnamefont {I.}~\bibnamefont {Bloch}}, \ and\ \bibinfo
  {author} {\bibfnamefont {M.}~\bibnamefont {Serbyn}},\ }\href {\doibase
  10.1103/RevModPhys.91.021001} {\bibfield  {journal} {\bibinfo  {journal}
  {Rev. Mod. Phys.}\ }\textbf {\bibinfo {volume} {91}},\ \bibinfo {pages}
  {021001} (\bibinfo {year} {2019})}\BibitemShut {NoStop}%
\bibitem [{\citenamefont {Decker}\ \emph {et~al.}(2020)\citenamefont {Decker},
  \citenamefont {Karrasch}, \citenamefont {Eisert},\ and\ \citenamefont
  {Kennes}}]{decker2020floquet}%
  \BibitemOpen
  \bibfield  {author} {\bibinfo {author} {\bibfnamefont {K.~S.~C.}\
  \bibnamefont {Decker}}, \bibinfo {author} {\bibfnamefont {C.}~\bibnamefont
  {Karrasch}}, \bibinfo {author} {\bibfnamefont {J.}~\bibnamefont {Eisert}}, \
  and\ \bibinfo {author} {\bibfnamefont {D.~M.}\ \bibnamefont {Kennes}},\
  }\href {\doibase 10.1103/PhysRevLett.124.190601} {\bibfield  {journal}
  {\bibinfo  {journal} {Phys. Rev. Lett.}\ }\textbf {\bibinfo {volume} {124}},\
  \bibinfo {pages} {190601} (\bibinfo {year} {2020})}\BibitemShut {NoStop}%
\bibitem [{\citenamefont {Chandran}\ \emph {et~al.}(2015)\citenamefont
  {Chandran}, \citenamefont {Kim}, \citenamefont {Vidal},\ and\ \citenamefont
  {Abanin}}]{chandran2015constructing}%
  \BibitemOpen
  \bibfield  {author} {\bibinfo {author} {\bibfnamefont {A.}~\bibnamefont
  {Chandran}}, \bibinfo {author} {\bibfnamefont {I.~H.}\ \bibnamefont {Kim}},
  \bibinfo {author} {\bibfnamefont {G.}~\bibnamefont {Vidal}}, \ and\ \bibinfo
  {author} {\bibfnamefont {D.~A.}\ \bibnamefont {Abanin}},\ }\href {\doibase
  10.1103/PhysRevB.91.085425} {\bibfield  {journal} {\bibinfo  {journal} {Phys.
  Rev. B}\ }\textbf {\bibinfo {volume} {91}},\ \bibinfo {pages} {085425}
  (\bibinfo {year} {2015})}\BibitemShut {NoStop}%
\bibitem [{ibm()}]{ibm_quantum_devices}%
  \BibitemOpen
  \href@noop {} {\enquote {\bibinfo {title} {{IBM} {Quantum} {Compute}
  {Resources}},}\ }\bibinfo {howpublished}
  {\url{https://quantum-computing.ibm.com/services}}\BibitemShut {NoStop}%
\end{thebibliography}
\end{document}